\documentclass[12pt]{article}

\usepackage{natbib}
\usepackage{amsmath}
\usepackage{amsfonts}
\usepackage{graphicx}
\usepackage{graphics}
\usepackage[margin=1in]{geometry}

\newtheorem{theorem1}{Theorem}
\newtheorem{ths1}{Theorem 1, Statement (i)}
\newtheorem{lemma}{Lemma}
\newtheorem{theorem}{Restriction}
\newtheorem{restriction}{Restriction}
\newtheorem{ths2}{Theorem 1, Statement (ii)}

\newenvironment{remark}[1][Remark]{\begin{trivlist}
\item[\hskip \labelsep {\bfseries #1}]}{\end{trivlist}}

\newenvironment{proof}[1][Proof]{\begin{trivlist}
\item[\hskip \labelsep {\bfseries #1}]}{\end{trivlist}}

\newenvironment{example}[1][Example]{\begin{trivlist}
\item[\hskip \labelsep {\bfseries #1}]}{\end{trivlist}}

\title{Half-Spectral Space-Time Covariance Models}
\author{Michael T. Horrell \& Michael L. Stein\\
\textit{University of Chicago}}

\begin{document}
\maketitle

\begin{abstract}
We develop two new classes of space-time Gaussian process models by specifying covariance functions using what we call a half-spectral representation.  The half-spectral representation of a covariance function, $K$, is a special case of standard spectral representations.  In addition to the introduction of two new model classes, we also develop desirable theoretical properties of certain half-spectral forms.  In particular, for a half-spectral model, $K$, we determine spatial and temporal mean-square differentiability properties of a Gaussian process governed by $K$, and we determine whether or not the spectral density of $K$ meets a regularity condition motivated by a screening effect analysis. We fit models we develop in this paper to a wind power dataset, and we show our models fit these data better than other separable and non-separable space-time models.

\smallskip

\end{abstract}

\section{Introduction }
Continuous natural phenomena in space or space-time are often modeled statistically as Gaussian processes. Generally, a Gaussian process model requires specification of a mean structure and a covariance or dependence structure. Since covariance models must be positive definite, mathematically valid dependence structures can be quite difficult to define. Therefore, a primary focus in defining new Gaussian process models is on defining valid covariance functions.  A covariance function, $K$, can be defined explicitly or implicitly in a number of ways.  For example, for continuous $K(x)$, where $K(x) = \int \exp(i \xi^Tx) F(d\xi)$ with $x, \xi \in \mathbb{R}^{d+1}$, specifying $F$ or the spectral density, $f$ (assuming it exists), where $F(d\xi) = f(\xi)d\xi$, will represent and characterize dependence structure in a Gaussian process.  In this paper, we explore an alternative representation of $K$ we call the half-spectrum.  We show half-spectral models are easy to define naturally, and we show many half-spectral models have several nice modeling properties.  

For a stationary, continuous space-time covariance model $K(s,t)$ with $s \in \mathbb{R}^d$ and $t \in \mathbb{R}$ with integrable spectral density, $g(\lambda,\omega)$, where $\lambda \in \mathbb{R}^d$ and $\omega \in \mathbb{R}$, we can represent $K$ with $g$ via Fourier transform,
\begin{equation}
K(s,t) = \frac{1}{(2 \pi)^{d+1}} \int_\mathbb{R} \int_{\mathbb{R}^d} g(\lambda,\omega) \exp\left(i s^T \lambda + i t \omega\right) d\lambda ~d\omega. \label{eqstart}
\end{equation}
A half-spectral representation of $K$ can be obtained by performing one of the two integration steps in (\ref{eqstart}) with respect to either $\lambda$ or $\omega$. After integrating over either $\lambda$ or $\omega$, $K$ is represented by the Fourier transform of an expression relating distance and frequency pairs: ($s,\omega$) or ($\lambda, t$).  We therefore call covariance representations of this type either half-spectral models in time (if in terms of ($s,\omega$)) or space (if in terms of ($\lambda,t$)). In this paper, we study half-spectra in time of the form $f(\omega)\mathbb{C}(s\delta(\omega)) \exp(i \theta(\omega) \phi^T s)$ defined such that
\begin{equation}
K(s,t) = \int_\mathbb{R} f(\omega) \mathbb{C}(s \delta(\omega)) \exp(i \theta(\omega) \phi^T s) \exp(i t \omega) d\omega \label{hs}.
\end{equation}
Here $f$ is the spectral density of $K(0,t)$, the temporal covariance function at a single spatial site; $\mathbb{C}$ is any valid correlation function in $\mathbb{R}^d$ with a spectral density; and $\delta$ is an even positive function that determines the space-time interaction properties of the model. The odd, real function, $\theta(\cdot)$, and the unit vector, $\phi \in \mathbb{R}^d$, determine space-time asymmetry in $K$ by specifying a phase relationship.  By Theorem 1 in \cite{regmondata}, we know $K$ thus defined is a valid covariance model. 

Space-time asymmetry can in many cases be thought of in terms of movement of a process through time.  \cite{gneiting} first introduced the term fully-symmetric to describe covariance models that have the following property: $K(s,t) = K(-s,t) = K(s,-t) = K(-s,-t)$ for all $s \in \mathbb{R}^d$ and $t \in \mathbb{R}$. For models of physical processes, full-symmetry can be a strong assumption; thus, $\theta$ and $\phi$ are important model building tools.  For much of this paper, we primarily consider the space-time symmetric case with $\theta(\omega) = 0$; hence, we mainly study half-spectra of the form $f(\omega)\mathbb{C}(s\delta(\omega))$. However, as we show in Section 2, some important modeling properties of $K$ will hold, independent of $\theta$ or $\phi$. Restriction of our analysis to models defined by $f(\omega)\mathbb{C}(s\delta(\omega))$ should therefore not be seen as a simplification.

\cite{creshua} first used half-spectral representations in space to help define several non-separable, closed-form space-time covariance functions.  \cite{gneiting} used half-spectral models as a stepping stone to extend much of the work in \cite{creshua}.  \cite{gneiting} also focused on establishing a large and flexible class of closed-form covariance functions.  \cite{regmondata} showed half-spectra in time can be fit directly to certain data structures using a multivariate Whittle likelihood approach.    \cite{arxiv} also focus on the class given in (\ref{hs}) and produce model fitting and inference strategies. Given the general usefulness of half-spectral models, it is natural to push their development. We go a step further than these authors by developing several properties of certain half-spectral forms.

The fully-symmetric, half-spectral form we study, $f(\omega)\mathbb{C}(s\delta(\omega))$, is quite flexible and lends itself to new interpretations.  Purely temporal properties of these models are determined by $f$, and purely spatial properties are determined by the real integral $\int f(\omega) \mathbb{C}(s \delta(\omega))d\omega$.  This real integral can also be seen as a mixture of covariance functions with varying ranges or simply as an expectation. These models also serve as one fairly direct way to extend common purely temporal models to the spatio-temporal domain.  The space-time interaction properties of these models are determined by $\delta$.  When $\delta$ is constant, temporal and spatial marginal covariance functions can be specified completely independently.  In this setting, $K(s,t)$ is proportional to a product of its spatial and temporal marginal covariance functions: $K(s,t) \propto K(s,0) K(0,t)$.  A model with this property is called a separable model, and like \cite{creshua}, \cite{gneiting} and several others, we avoid separable models because they have been criticized as unnatural \citep{kyriakidis}.  Finally, if the spectrum of $\mathbb{C}$ is known, the full spectral representations of $f(\omega)\mathbb{C}(s\delta(\omega))$ or $f(\omega)\mathbb{C}(s\delta(\omega)) \exp(i \theta(\omega) \phi^T s)$ can be determined straightforwardly; thus, full spectral diagnostics and considerations can be easily applied to half-spectral models of this type.

We use the flexibilities of $f(\omega)\mathbb{C}(s\delta(\omega))$ and $f(\omega)\mathbb{C}(s\delta(\omega)) \exp(i \theta(\omega) \phi^T s)$ to build new classes of models with several desirable modeling properties.  \cite{steinbook} has established the tremendous importance of the degree of mean-square differentiability (or smoothness) of a process on predictions and assessing uncertainty in these predictions.  We therefore work to develop models that have flexible degrees of smoothness in space and time.  We also consider carefully the space-time interaction properties of the models we develop.  Again, we avoid space-time separability, but we additionally attempt to define models that have ``natural'' space-time interaction properties. There is less research on properties of space-time interaction in space-time models; hence, ``natural'' space-time interaction is difficult to define without a specific application in mind. However, \cite{steinscreen} provides one condition that we might want to assume in practice.  Specifically, \cite{steinscreen} argues for the following condition on the full spectrum, $g$:  
\begin{equation}
\lim_{\|(\lambda,\omega)\|\rightarrow \infty} \sup_{\|(u,v)\| < R} \left| \frac{g(\lambda + v,\omega + u)}{g(\lambda,\omega)} - 1 \right| = 0 \label{cond}
\end{equation}
for all finite $R$. The limit $\|(\lambda,\omega)\| \rightarrow \infty$ indicates any path to infinity may be taken.  In particular, either $\lambda$ or $\omega$ may be bounded or fixed as the other coordinate goes to infinity.  Essentially, the condition in (\ref{cond}) requires $g$ to flatten out in a relative sense at high values of $\lambda$ and/or $\omega$. 

\cite{steinscreen} justifies (\ref{cond}) by considering interpolation properties of models determined by $g$.  Specifically, \cite{steinscreen} considers asymptotic behavior of best linear unbiased interpolations using points near to and far from an interpolation site.  \cite{steinscreen} proves that if $g$ satisfies (\ref{cond}) and has further mild properties, then as the nearby points approach the interpolation site, interpolation using only the information in the nearby points becomes optimum in the sense that adding the information from the far points does not improve predictions.  This broad concept is called the screening effect \citep{chiles}. The theorems in \cite{steinscreen} are not completely general, but they suggest that a space or space-time model with a spectral representation that does not follow the condition in (\ref{cond})  may have undesirable interpolation properties.  We therefore use this condition to guide development of our classes of half-spectral models.  Hence, the models we develop here may have more natural interpolation properties and, by extension, more natural space-time interaction.

As further confirmation that the condition in (\ref{cond}) may describe natural models, a number of physically derived Gaussian process models will satisfy (\ref{cond}).  For example, models defined using stochastic partial differential equations will often meet (\ref{cond}) when white noise forcing terms lead to processes with integrable spectral densities.  In particular, Gaussian process solutions to stochastic versions of convection-diffusion equations will satisfy (\ref{cond}) \citep{spde}.  

Several classical  ``unnatural'' models do not meet the condition in (\ref{cond}).  Separable models in particular do not follow (\ref{cond}).  In one or more dimensions, the squared exponential covariance, $K(t) = \exp(-t^2)$ which produces analytic realizations is also excluded since its spectral density, $\sqrt{\pi} \exp(-\omega^2/4)$,  goes to zero too quickly to satisfy (\ref{cond}).  The triangular covariance function $K(t) = (1 - |t|)^+$ which has documented undesirable interpolation properties is excluded as well since its spectral density, $2(1- \cos(\omega))/\omega^2$, oscillates too violently to flatten out in the relative sense required by (\ref{cond}) \citep{steinbook}.

It should also be noted the condition in (\ref{cond}) is a more general condition than one posed in \cite{steinold}, wherein $g$ is required to be asymptotically proportional to a regularly varying function.  Like \cite{steinscreen}, \cite{steinold} uses the screening effect to guide generation of restrictions on $g$.  We use (\ref{cond}) over the condition in \cite{steinold} because it describes a more general class of functions.  Additional screening effect analysis can be found in \cite{steinnew}, wherein conditions showing when a screening effect will not hold are given.  \cite{steinnew} also supports use of models satisfying exactly the condition in (\ref{cond}).

Beyond half-spectral methods, there are a host of other techniques that can be used to produce valid non-separable space-time and other multidimensional covariance models.  \cite{blur} present a convolution based method to produce a flexible anisotropic class of models. \cite{ma} presents a mixture type method to combine marginal covariances into a nonseparable model.  \cite{kolovos} provide an overview of many methods including full spectral methods and methods requiring solution of stochastic partial differential equations.  Furthermore, \cite{kolovos} present a technique to produce models valid in many dimensions from one-dimensional models via weighted sums.  \cite{genton} approximates non-separable models using separable covariances.  \cite{fuentes} and \cite{steincov} use full spectral representations to define non-separable models.

In this paper, we advance the half-spectral covariance model representation in (\ref{hs}) and show that only specific forms of functions $\mathbb{C}$, $f$, $\delta$ and $\theta$ will lead to models satisfying the condition in (\ref{cond}).  In the next section of this paper, we show space-time models will satisfy (\ref{cond}) independent of the specification of $\theta$ as long as $\theta$ is locally bounded.  This conclusion allows our remaining analysis in this section to focus on fully-symmetric models.  We develop two restrictions on the set of functions, $\mathbb{C}$, $f$ and $\delta$ and prove that models violating these restrictions will not meet (\ref{cond}).  These restrictions are subsequently used in Section 3 to guide development of two new classes of half-spectral models. One class of models satisfies both restrictions presented in Section 2.  The second class satisfies only one of the restrictions in Section 2.  For both classes, we show how smoothness for a fully symmetric model in space and time may be changed by specification of components of $\mathbb{C}$, $f$ and $\delta$.  In Section 4, we fit examples from both classes of models to daily average wind speed data recorded at different spatial locations across Ireland.  We compare fits of these models to other models of similar form. We find models in the classes we develop to be easily adaptable to the data we consider, and we show they better capture space-time interaction in the Irish wind dataset compared to other separable and other non-separable models.

\section{Restrictions on half-spectral forms}

The half-spectral form in (\ref{hs}) in general does not meet the condition in (\ref{cond}).  Let $h$ be the spectral density of the spatial covariance function $\mathbb{C}$. The full spectral representation of $K$ is
\begin{equation}
K(s,t) = \int_\mathbb{R} \int_{\mathbb{R}^d} f(\omega) \delta(\omega)^{-d} h\left(\frac{\lambda-\theta(\omega) \phi}{\delta(\omega)}\right) \exp(i s^T \lambda + i t \omega ) d\lambda ~ d\omega. \label{fullSpec}
\end{equation} 
Let $g(\lambda,\omega)$ be the fully-symmetric spectrum, $g(\lambda,\omega) =  f(\omega) \delta(\omega)^{-d} h\left(\lambda/\delta(\omega)\right)$; therefore, we can write the general spectrum of $K$  as $g(\lambda-\theta(\omega)\phi,\omega)$.  To show when $g(\lambda,\omega)$ or $g(\lambda-\theta(\omega)\phi,\omega)$ satisfy (\ref{cond}), we first establish the result in Theorem 1 that states: $g(\lambda,\omega)$ satisfies (\ref{cond}) if and only if $g(\lambda-\theta(\omega)\phi,\omega)$ satisfies (\ref{cond}) as long as a mild regularity condition requiring $\theta(\omega)$ to be locally bounded is satisfied.  This result permits the theoretical focus of this section to center on fully-symmetric models without much reducing the generality of our results.  With Theorem 1 established, the remaining theory in this section gives necessary conditions and restrictions on a fully-symmetric model (defined only by $f$, $\mathbb{C}$ and $\delta$) that must be met in order for $g$ to satisfy (\ref{cond}).  Proofs of Theorem 1 and the Restrictions presented in this section are in the Appendix.

\begin{theorem1}
Let $g(\lambda,\omega)$ be the full spectrum of a space-time covariance function where $\lambda$ is the spatial wavenumber and $\omega$ is the temporal frequency. Let $\theta(\cdot)$ be an odd function, and let $\phi$ be a unit vector in $\mathbb{R}^d$. The following two statements hold:
\begin{itemize}
\item[(i)] Let $\theta(\cdot)$ be locally bounded.  The full spectrum $g(\lambda,\omega)$ satisfies (\ref{cond}) if and only if \newline $g(\lambda - \theta(\omega) \phi,\omega)$ satisfies (\ref{cond}).  
\item[(ii)] Let $g(\lambda,\omega)$ be strictly positive and let $g(\lambda,\omega)$ satisfy (\ref{cond}). If $g(\lambda - \theta(\omega) \phi,\omega)$ satisfies (\ref{cond}), then $\theta(\omega)$ must be locally bounded.
\end{itemize}
\label{THM1}
\end{theorem1}

\begin{remark}
Statement (i) is bolstered in a practical sense by statement (ii).  Model building strategies often start with a fully-symmetric model that may be tweaked to be not fully symmetric. If we accept (\ref{cond}) as a desirable condition, both the fully-symmetric and the not fully-symmetric versions of a model should satisfy (\ref{cond}). Statement (ii) therefore establishes that the regularity condition in statement (i) (the local boundedness of $\theta(\omega)$) is a natural condition.
\end{remark}

With Theorem 1 in mind, we focus on fully-symmetric half-spectra $f(\omega)\mathbb{C}(s\delta(\omega))$ with full spectra $g(\lambda,\omega)$.  For $g(\lambda,\omega)$ to meet the condition in (\ref{cond}), we show the spectrum $h$  and function $\delta$ must lead to a factorizable form: $h(\lambda/\delta(\omega)) = H(\lambda, \omega) \delta(\omega)^d/f(\omega) $, where $H$ satisfies (\ref{cond}).  Though $h$ and $\delta$ of this form can be seen to trivially meet (\ref{cond}), the restrictions developed in the following paragraphs establish certain more flexible forms will not satisfy (\ref{cond}). 

Since $g(\lambda,\omega) = f(\omega) \delta(\omega)^{-d} h\left(\lambda/\delta(\omega)\right)$, we can more simply write $g$ in terms of a marginal function in $\omega$ and a joint function in $\lambda$ and $\omega$. Let $g(\lambda,\omega) = p(\omega) q(\lambda,\omega)$.  The following restrictions establish necessary conditions on $p$ and $q$ that must be met in order for $g$ to satisfy the condition in (\ref{cond}).

\begin{theorem}
Let the spectral representation of covariance function $K$ have the form $g(\lambda,\omega) = p(\omega) q(\lambda,\omega)$.  The following two statements hold for spectral densities with this parameterization.
\begin{itemize}
\item[(i)] If $g$ and $q$ satisfy the condition in (\ref{cond}), then $p$ must be constant in $\omega$.
\item[(ii)] If $g$ satisfies the condition in (\ref{cond}), then for every point $\omega_0$ such that $p(\omega_0)$ is finite and positive, the marginal spectrum $g^*_{\omega_0}(\lambda) = g(\lambda,\omega_0)$ must meet the condition in (\ref{cond}).
\end{itemize}
\label{t1}
\end{theorem}

\begin{remark}
The requirement in (i) that $q$ follows the condition in (\ref{cond}) may be weakened. Our proof for (i) needs only that $\lim q(\lambda, \omega_0+u_0)/q( \lambda, \omega_0) \neq 1/C$, a fact which follows from $q$ satisfying (\ref{cond}).  This condition however seems quite arbitrary, and we have opted to use the stricter condition.  
\end{remark}

The implication of Restriction \ref{t1} is simply that the space-time interaction function $\delta$ cannot be determined completely independently of $\mathbb{C}$ and $f$ if we wish to use models that follow the condition in (\ref{cond}).  Of note is the fact that this restriction applies to several models proposed in \cite{creshua}, \cite{gneiting} and \cite{regmondata}.

One may define the spatial interaction function as $\delta(\omega) = f(\omega)^{1/d}$ to escape the conditions of Restriction \ref{t1}; however, $\delta(\omega)$ defined in this way decreases to zero for increasingly large $\omega$ since $f$ must be integrable.  This behavior of $\delta$ contradicts intuition---for example, considering (\ref{hs}), we see $\delta(\omega) \rightarrow 0$ implies coherences between multiple time series are highest at high frequencies.  More formally, $\delta$ decreasing to zero or having a finite limit as $\omega$ approaches infinity will not generally lead to models satisfying the condition in (\ref{cond}).

\begin{theorem}
Let $g(\lambda,\omega) = p(\omega)q(\lambda/\delta(\omega))$, where $p(\omega)$ is positive and finite for all sufficiently large $|\omega|$, $q$ is a non-negative, continuous, integrable function that is not identically zero, and $\delta(\omega)$ is even with a well-defined limit as $|\omega| \rightarrow \infty$.  If $g$ satisfies (\ref{cond}),  then $\lim_{|\omega| \rightarrow \infty} \delta(\omega) = \infty$. \label{t2}
\end{theorem}

\begin{remark}
A version of this Restriction exists concerning $\delta(\omega)$ without a well-defined limit as $|\omega| \rightarrow \infty$.  If  $\liminf_{|\omega|\rightarrow\infty} \delta(\omega)$ is finite, then under the additional condition that $q$ satisfies (\ref{cond}) it can be proved that $g$ cannot satisfy (\ref{cond}).  We let $q$ be more general and include the stronger condition on $\delta$ in this Restriction to match the class of models we develop in Section 3.2.
\end{remark}

Restrictions \ref{t1} and \ref{t2} confine the form of $\delta$.  Implicitly, the form of $h$ is also restricted. For half-spectral models satisfying (\ref{cond}), if $\delta$ must increase to infinity, $f(\omega) \delta(\omega)^{-d}$ cannot be constant, implying $h(\lambda/\delta(\omega))$ itself cannot satisfy (\ref{cond}).  On the other hand, for fixed $\omega$, $h(\lambda/\delta(\omega))$ must satisfy (\ref{cond}) in $\lambda$ by part (ii) of Restriction \ref{t1}.  One way forward is to choose $h$ and $\delta$ such that $h(\lambda/\delta(\omega))$ admits the representation, $h(\lambda/\delta(\omega)) =  H(\lambda,\omega) \delta(\omega)^d/f(\omega)$, where $H$ satisfies (\ref{cond}).  A half-spectral model thus specified will satisfy (\ref{cond}). Finally, by Theorem 1, with any locally bounded $\theta$, $h((\lambda - \theta(\omega)\phi)/\delta(\omega))$ specified in this way will also lead to models satisfying (\ref{cond}).

\section{New classes of models}

We build two new classes of fully-symmetric models; thus $\theta(\omega) = 0$ for each of the classes of models we present in this section.  If $\theta(\omega)$ is non-zero, models with space-time asymmetry can be defined.  Importantly, as we showed in Section 2, a zero or non-zero $\theta(\omega)$ does not affect whether or not a given model will satisfy (\ref{cond}) as long as $\theta$ is a locally bounded function.

\subsection{A class of models satisfying the natural condition}

Let $h(\lambda) = \phi (\alpha^2 + \|\lambda\|^2)^{-(\nu + d/2 + 1/2)}$, the Mat\'{e}rn spectral density with smoothness $\nu + 1/2 > 0$, $\alpha > 0$ an inverse range parameter and $\phi > 0$, a variance parameter.  For a given $f$, to ensure Restrictions \ref{t1} and \ref{t2} are not violated, choose $\delta(\omega) = f(\omega)^{-1/2 \frac{1}{\nu + 1/2}}$. Plugging into the space-time spectrum in (\ref{fullSpec}), a new set of spectral densities can be written
\begin{equation}
g(\lambda,\omega) = \phi\left(\alpha^2 f(\omega)^{-1/(\nu + 1/2)} + \|\lambda\|^2  \right)^{-(\nu+(d+1)/2)}. \label{hsf}
\end{equation}
If $f$, itself a spectral density, satisfies $f \sim c_f \omega^{-k}$ for some $k > 1$ and $c_f$ a positive constant depending on $f$, then (\ref{hsf}) describes a class of models that satisfies (\ref{cond}).  We write covariance functions with this half-spectrum $K_f$ to indicate the temporal spectrum $f$ determines a specific model in this class.  

We use a Mat\'{e}rn spectral density in space to form this class of models for three reasons.  First, it is a natural model to analyze since spatial modeling is often carried out using a Mat\'{e}rn covariance or a modification of a Mat\'{e}rn covariance.  Second, its form leads to easy comparison with the conditions of Restrictions \ref{t1} and \ref{t2}.  Third, the Mat\'{e}rn spectrum gives (\ref{hsf}) some desirable Mat\'{e}rn-like properties.  For example, as we show in the following paragraphs, the parameter $\nu$, is directly related to the spatial smoothness of the resulting Gaussian process.  The 1/2 added to $\nu$ ensures this formulation extends the space-time Mat\'{e}rn covariance naturally.  If $f(\omega)$ is a one-dimensional Mat\'{e}rn spectral density with smoothness $\nu$, then $K_f(s,t)$ is simply a $d + 1$ dimensional Mat\'{e}rn covariance function with smoothness $\nu$.  Note, in this case, $\nu > 0$, since $f$ must be integrable.

The half-spectral representation of (\ref{hsf}) is given simply by plugging the Mat\'{e}rn covariance into the relevant portion of (\ref{hs}).
\begin{eqnarray}
f(\omega) \mathbb{C}(s \delta(\omega)) &=&  f(\omega) \frac{\pi^{d/2} \phi}{2^{\nu - 1/2} \Gamma(\nu + (d+1)/2) \alpha^{2\nu + 1}} \times \nonumber \\
& & (\alpha \|s\|  f(\omega)^{-1/2 \frac{1}{\nu + 1/2}})^{v+1/2} \mathcal{K}_{\nu + 1/2}  \left(\alpha \|s\|  f(\omega)^{-1/2 \frac{1}{\nu + 1/2}}\right) \label{HSFORM1}
\end{eqnarray}
where $\mathcal{K}_{\nu+1/2}$ is the modified Bessel function of the second kind with argument $\nu+ 1/2$ \citep{nist}.  If we let $\phi$ absorb terms not depending on $\omega$ in (\ref{HSFORM1}) and write $\mathcal{M}_{\nu + 1/2}(s) = \|s\|^{\nu+1/2} \mathcal{K}_{\nu+1/2}(\|s\|)$, (\ref{HSFORM1}) can be simplified to 
\begin{equation}
f(\omega) \mathbb{C}(s \delta(\omega)) = \phi f(\omega) \mathcal{M}_{\nu+1/2} \left(\alpha \|s\|  f(\omega)^{-1/2 \frac{1}{\nu + 1/2}} \right). \label{HSFORM2}
\end{equation}

Beyond satisfying the condition in (\ref{cond}), mean-square differentiable properties of this class of models are readily found.  For a space-time process, $Z(s,t)$, governed by stationary covariance $K(s,t)$, $Z(s,t)$ is $m$ times mean-square differentiable in time if and only if $\frac{\partial^{2m}}{\partial t^{2m}}K(0,0)$ exists \citep{steincov}. Following \cite[p. 5]{steinweiss}, $k$-th derivatives of $K(0,0)$ in the temporal direction are known to exist if the purely temporal spectral density, $f(\omega)$, is integrable and $\int f(\omega) |\omega|^k < \infty$.  Hence, if $f$ integrable and $f \sim c_f |\omega|^{-k}$ as $\omega$ approaches infinity, then $Z(s,t)$ will be $m$ times mean-square differentiable in time if $(k-1)/2 > m$.  Similar statements hold regarding the mean-square differentiability of $Z(s,t)$ in space. To determine spatial mean-square differentiability, we mirror the argument in \cite[p. 313]{steincov}, showing $Z(s,t)$ will be $m$ time mean-square differentiable in space if and only if $(\nu + 1/2)\cdot (k - 1)/k > m$.  Therefore, with these models, smoothnesses in space and time can be defined via $\nu$ and the limiting behavior of $f$ determined by $k$.  We call $(\nu+1/2) \cdot (k-1)/k$ and $(k-1)/2$ the effective smoothnesses of the purely spatial and purely temporal processes respectively.  Since $\nu > -1/2$ and $k > 1$, arbitrary effective smoothnesses in space and time can be specified using models in this class.

\begin{example}{\bf 1 (Mat\'{e}rn in Time)}

Let $f(\omega) = (\beta^2+\omega^2)^{-(\kappa + 1/2)}$, a one dimensional Mat\'{e}rn spectral density.  To avoid overparameterization, the scale parameter is fixed at one.  Plugging $f$ into (\ref{hsf}) and (\ref{HSFORM2}), we get the spectral and half-spectral forms.
\begin{eqnarray}
g(\lambda,\omega) &=&  \phi\left(\alpha^2  (\beta^2+\omega^2)^{(\kappa + 1/2)/(\nu + 1/2)} + \lambda^2  \right)^{-(\nu+(d+1)/2)}; \nonumber \\
f(\omega) \mathbb{C}(s \delta(\omega))&=&   \phi \left( \beta^2+\omega^2\right)^{-(\kappa + 1/2)}\mathcal{M}_{\nu+1/2} \left(\alpha \|s\|   (\beta^2+\omega^2)^{1/2\frac{\kappa + 1/2}{\nu + 1/2}}\right). \label{mit}
\end{eqnarray}
Here $k = 2 \kappa + 1$; hence this process is $\lceil \kappa \frac{\nu + 1/2}{\kappa + 1/2} - 1\rceil$ times mean-square differentiable in space and $\lceil \kappa-1 \rceil$ times mean-square differentiable in time, where $\lceil \cdot \rceil$ is the ceiling function. Again, it is readily checked that if $\kappa = \nu$, $K_f$ reduces to a $d+1$ dimensional Mat\'{e}rn model with perhaps different range parameters in time and space, depending on $\beta$ and $\alpha$. A numerically calculated contour plot for $\nu = 1/8$, $\kappa = 2$ is in Figure \ref{MTcovs}(a).  Marginal plots of $K_f(0,s)$ and $K_f(t,0)$ are in Figure \ref{MTcovs}(b).
\begin{figure}
\centering
(a)
\begin{minipage}[t]{0.45\linewidth}
\centering
\includegraphics[scale=.295]{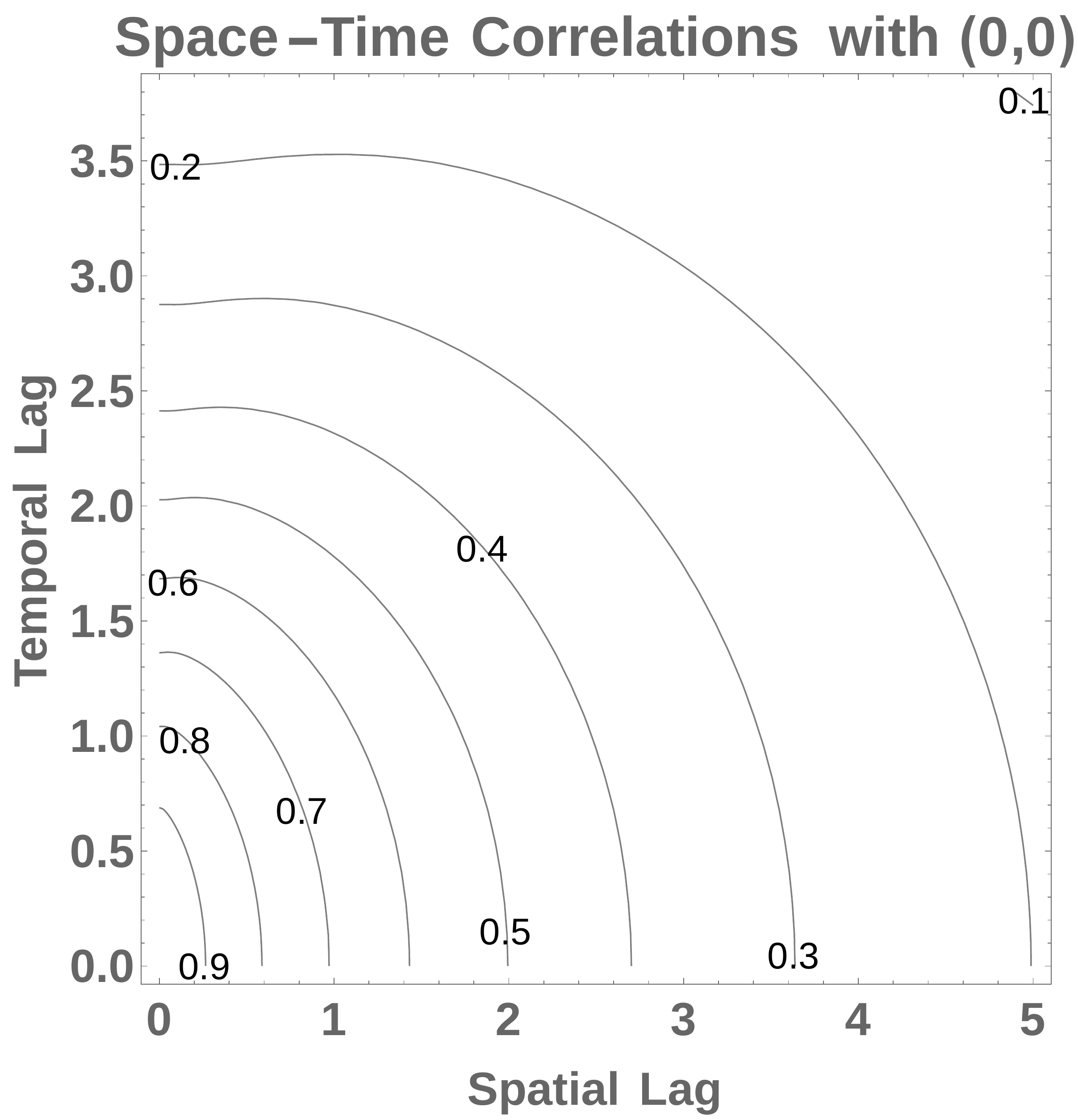}
\end{minipage}
(b)
\begin{minipage}[b]{0.45\linewidth}
\centering
\includegraphics[scale=.27]{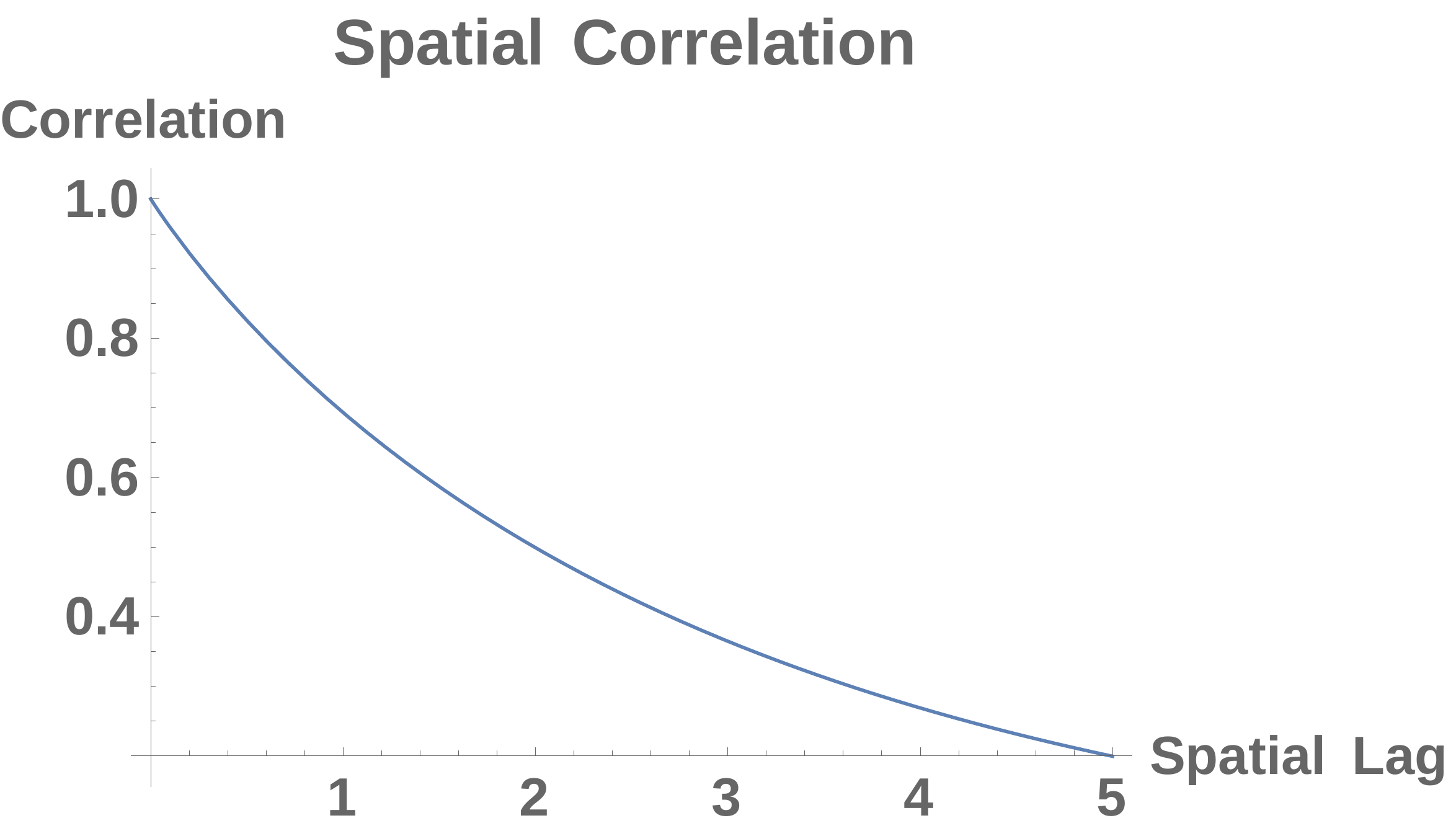}  \\ \includegraphics[scale=.27]{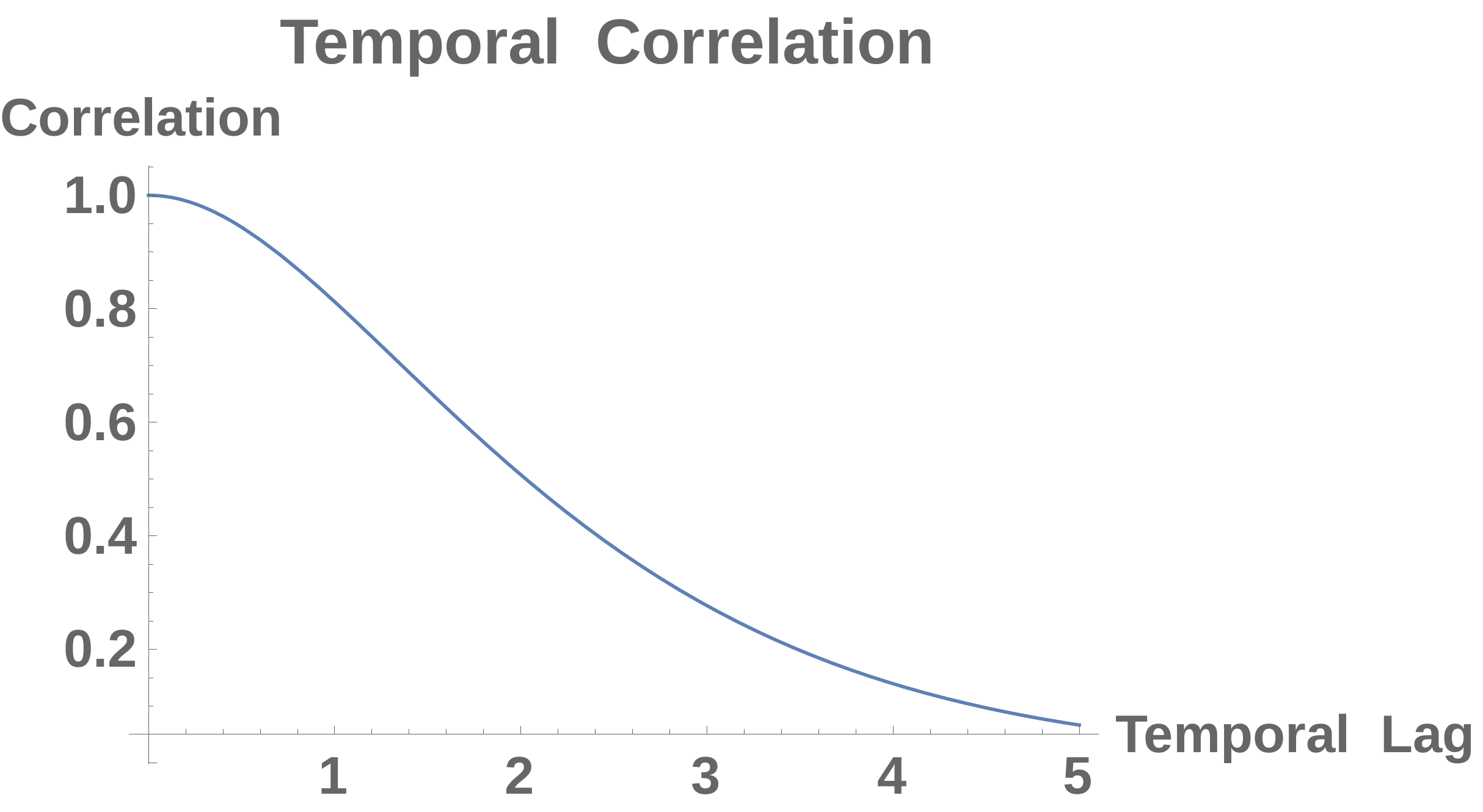}
\end{minipage}
\caption{Contour and marginal covariance plots for the Mat\'{e}rn in Time model in Example 1.  Effective smoothness in space is 0.5.  Effective smoothness in time is 2.}
\label{MTcovs}
\end{figure}

Of note in Figure \ref{MTcovs}(a) is the fact that a ``dimple''-type property is present in this model \citep{kent}.  Though the dimple we see in Figure \ref{MTcovs}(a) may not strictly follow the dimple definition given in \cite{kent}, the general phenomenon---that a dimple is indicative of higher covariances existing at non-zero spatial lags when the temporal lag is nonzero---is clearly present in this model.  A broader definition of a dimple may be the following: a dimple exists if a contour of a covariance function  is a non-convex shape.  Empirically, this appears to be caused by the large difference in $\kappa$ and $\nu$ and $\nu < \kappa$.  For $\kappa$ and $\nu$ closer together or $\nu > \kappa$, numerical experiments suggest this dimple either disappears or gets pushed further from the origin.  Since half-spectral forms in space can be defined similar to (\ref{mit}), models without apparent dimples with flexible smoothnesses in space and time can be defined.  Note, this model is a special case of a class of models considered by \cite{steincov}, but the half-spectral form of this model is not given in \cite{steincov} nor is the connection of this model to Mat\'{e}rn models mentioned therein.

\end{example}

\begin{example}{\bf 2 (Continuous AR(2) in Time)}

From \cite[p. 239]{Priestley}, a continuous analog of an auto-regressive process of order 2 can be defined by the following two parameter spectrum.
\begin{equation}
f(\omega) = \frac{1}{\pi}\left( \frac{\beta_1 \beta_2}{(\beta_2 - \omega^2)^2 + \beta_1^2 \omega^2} \right). \nonumber
\end{equation}
For large $\omega$, $f(\omega) \sim c_f \omega^{-4}$; hence this process is one time mean-square differentiable in time.  Contour and marginal plots are in Figure \ref{ARcovs}. 

\begin{figure}
\centering
(a)
\begin{minipage}[t]{0.45\linewidth}
\centering
\includegraphics[scale=.31]{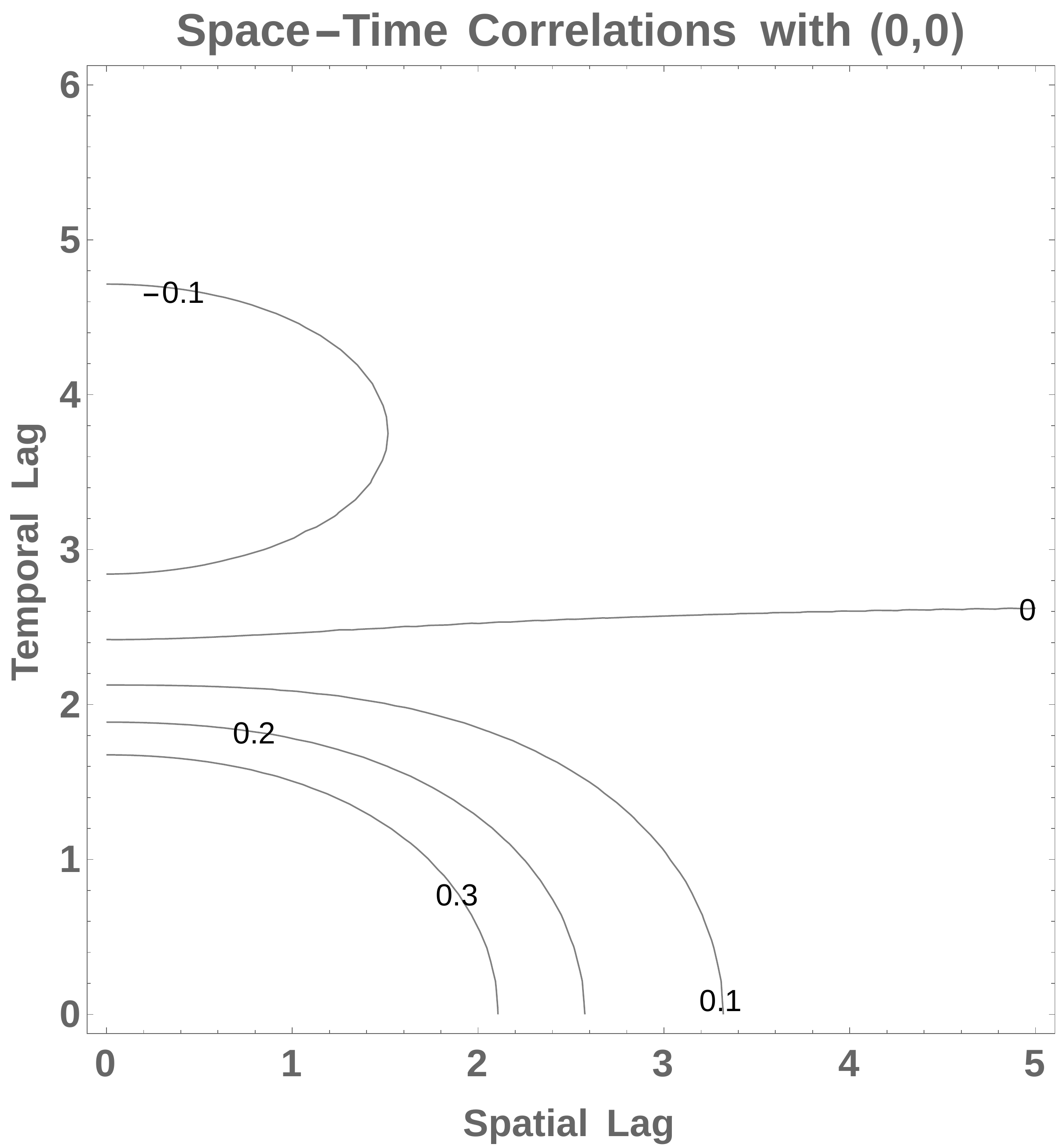}
\end{minipage}
(b)
\begin{minipage}[b]{0.45\linewidth}
\centering
\includegraphics[scale=.27]{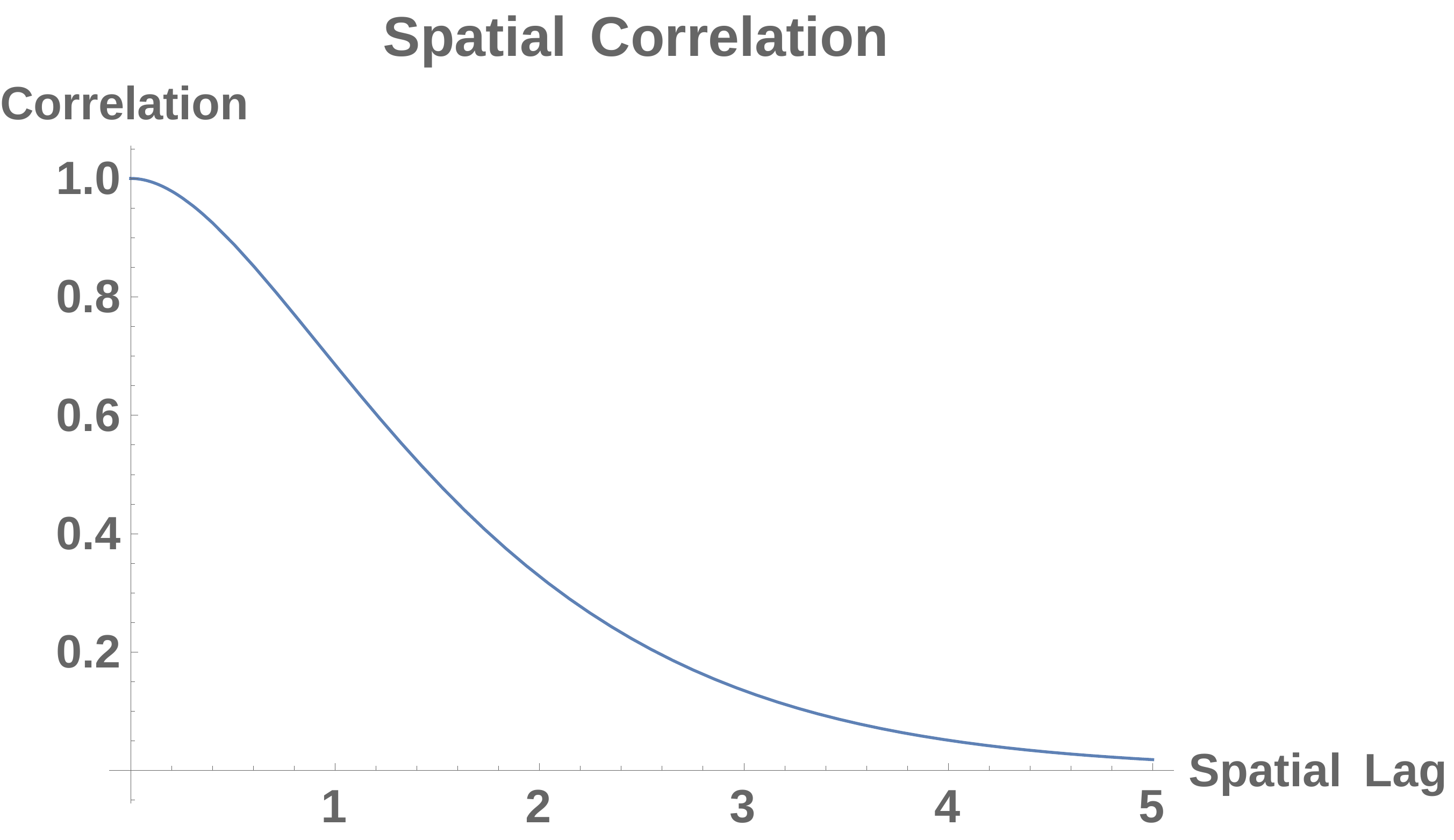}  \\ \includegraphics[scale=.25]{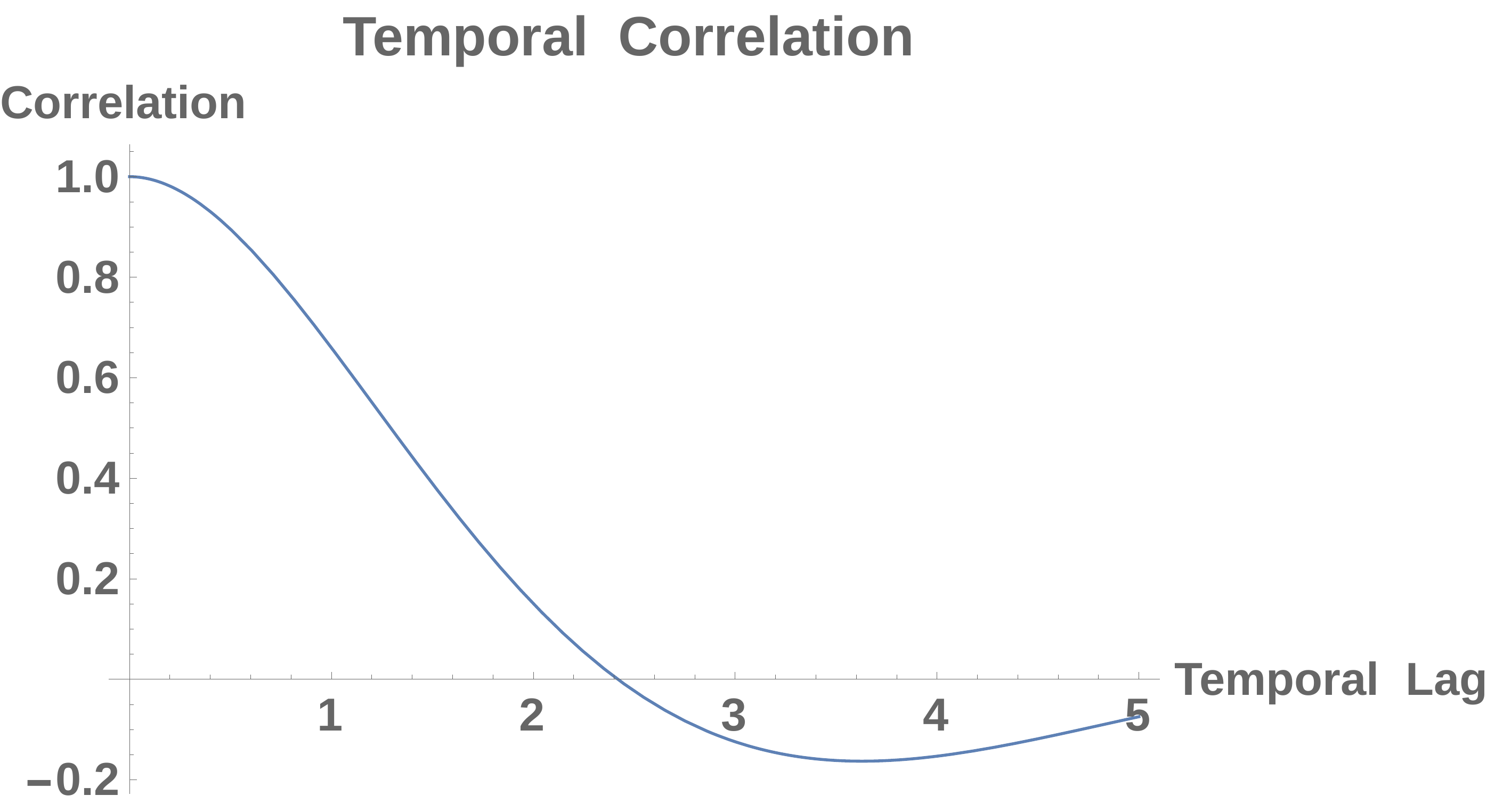}
\end{minipage}
\caption{Contour and marginal covariance plots for the AR(2) in Time model in Example 2.  Effective smoothnesses in space and time are each 3/2.  Auto-regressive parameters $\beta_1 = \beta_2 = 1$.}
\label{ARcovs}
\end{figure}

The oscillatory behavior of an AR(2) process is documented in Figure \ref{ARcovs}(a) and (b) by the negative covariances.  There is a hint of a dimple in Figure \ref{ARcovs}(a), though the effect does not appear to be severe.  Like the Mat\'{e}rn in Time model in Example 1, the dimple effect becomes more pronounced when smoothness in space is much lower than smoothness in time, and it either disappears or becomes unnoticeable when smoothness in space is much larger than smoothness in time.  
 
\end{example}

%

%


\subsection{A class of half-spectral models satisfying Restriction 2}

Restriction 2 effectively shows that a model of the form (\ref{hs}) must satisfy $\delta(\omega) \rightarrow \infty$ as $\omega \rightarrow \infty$ in order for (\ref{cond}) to hold.  Of course, models may satisfy $\delta(\omega) \rightarrow \infty$ as $\omega \rightarrow \infty$ without satisfying (\ref{cond}).  We consider $\delta \rightarrow \infty$ as $\omega \rightarrow \infty$ a condition of interest in its own right.  In particular, it forces coherences between time series at distinct spatial sites to tend to 0 as $\omega \rightarrow \infty$.  In this Section, we define a class of models satisfying $\delta(\omega) \rightarrow \infty$ as $\omega \rightarrow \infty$ but not (\ref{cond}). 

Briefly, if $f$ is a probability density (letting $\mathbb{C}$ be more generally a covariance function), we may write
\begin{equation}
K(s,0) = \mathbb{E}_\Omega\left[\mathbb{C}(s \delta(\Omega))\right] \label{HSfalse}
\end{equation}
where $\Omega$ is distributed according to $f$.  Hence, $K(s,0)$ depends only on the distribution of $\delta(\Omega)$.  Note further that since $K(0,t)$ is determined completely by $f$, $K(0,t)$ depends solely on the distribution of $\Omega$.  Since $f$ is even, $\Omega$ must be a symmetric random variable, and since $\delta(\omega)$ is even and positive, $\delta(\Omega)$ is distributed as the absolute value of a symmetric random variable.  Let $A$ and $B$ be continuous, symmetric random variables.  We can specify $\Omega \sim B$ and $\delta(\Omega) \sim |A|$; thus, $\delta(\omega) = |F_A^{-1}(F_B(\omega))|$ for $F_A$ and $F_B$ distribution functions. Therefore, in principle, $K(s,0)$ and $K(0,t)$ can be defined independently without using constant $\delta$.  Moreover, $\delta$ of this form increases unboundedly; hence this set of models does not violate Restriction 2.

Similar to a class of models defined in \cite{ma}, we further define this class by letting $\mathbb{C}(s) = \exp(-\|s\|^2)$.  With this $\mathbb{C}$, the spatial model in (\ref{HSfalse}) is a class of functions studied in \cite{schoenberg}.  If we allow $|A|$ to follow an arbitrary positive probability distribution on $\mathbb{R}$, \cite{schoenberg} proved spatial models in (\ref{HSfalse}) with $\mathbb{C}$ squared exponential characterize all covariance functions that are valid in all dimensions.  In other words, the class of models (\ref{HSfalse}) can marginally produce is quite vast.  Nonetheless, $\mathbb{C}$ as squared exponential has a spectrum proportional to $\exp(-\|\lambda\|^2/4)$; therefore, by Restriction \ref{t1} (ii), these models will not meet the condition in (\ref{cond}).

\begin{example} {\bf 3 (Marginal Mat\'{e}rn)}

Let $f = f_B = \beta^{2\kappa} \frac{\Gamma(\kappa + 1/2)}{\sqrt{\pi}\Gamma(\kappa)} (\beta^2 + \omega^2)^{-(\kappa + 1/2)}$ with $\kappa$, $\beta > 0$.  Let $\mathbb{C}(s) = \phi  \exp(-\alpha^2 \|s\|^2)$, and let $A$ be defined symmetrically by $|A| \sim 1/\sqrt{2 \chi^2_{2\nu}}$, where $\nu$, $\phi$, $\alpha > 0$.  It can be shown via results in \cite[p. 146]{bateman} that $K(s,0) = \phi/(2^{\nu - 1} \Gamma(\nu)) \mathcal{M}_{v}(\alpha s)$.  Since $f_B$ is a Mat\'{e}rn spectral density, $K(0,t) = \phi/(2^{\kappa - 1} \Gamma(\kappa)) \mathcal{M}_{\kappa}(\beta t)$. In Figure \ref{MarMatcovs}, we give marginal and contour plots comparable to those in Figure \ref{MTcovs}.
\begin{figure}
\centering
(a)
\begin{minipage}[t]{0.45\linewidth}
\centering
\includegraphics[scale=.35]{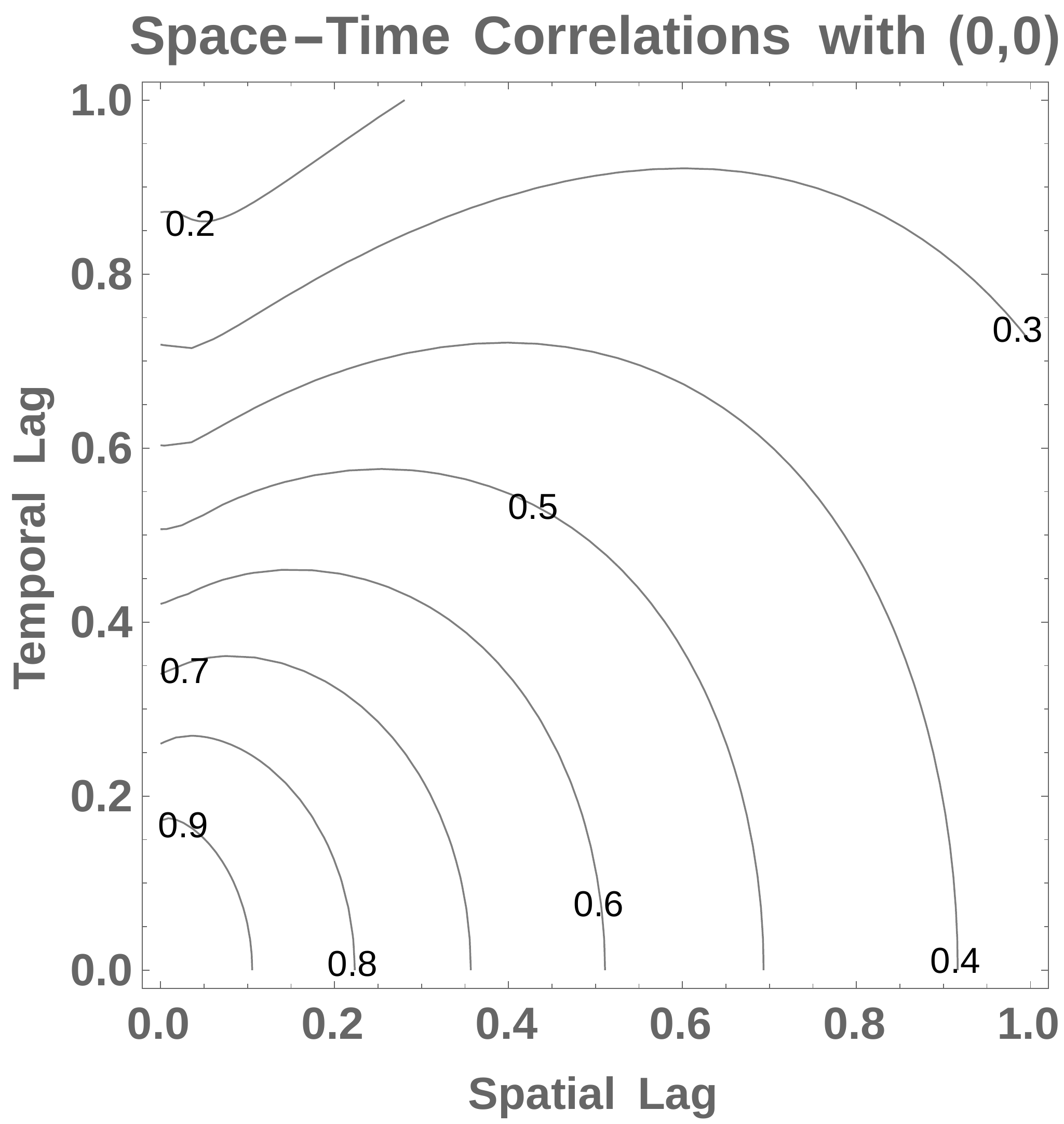}
\end{minipage}
(b)
\begin{minipage}[b]{0.45\linewidth}
\centering
\includegraphics[scale=.28]{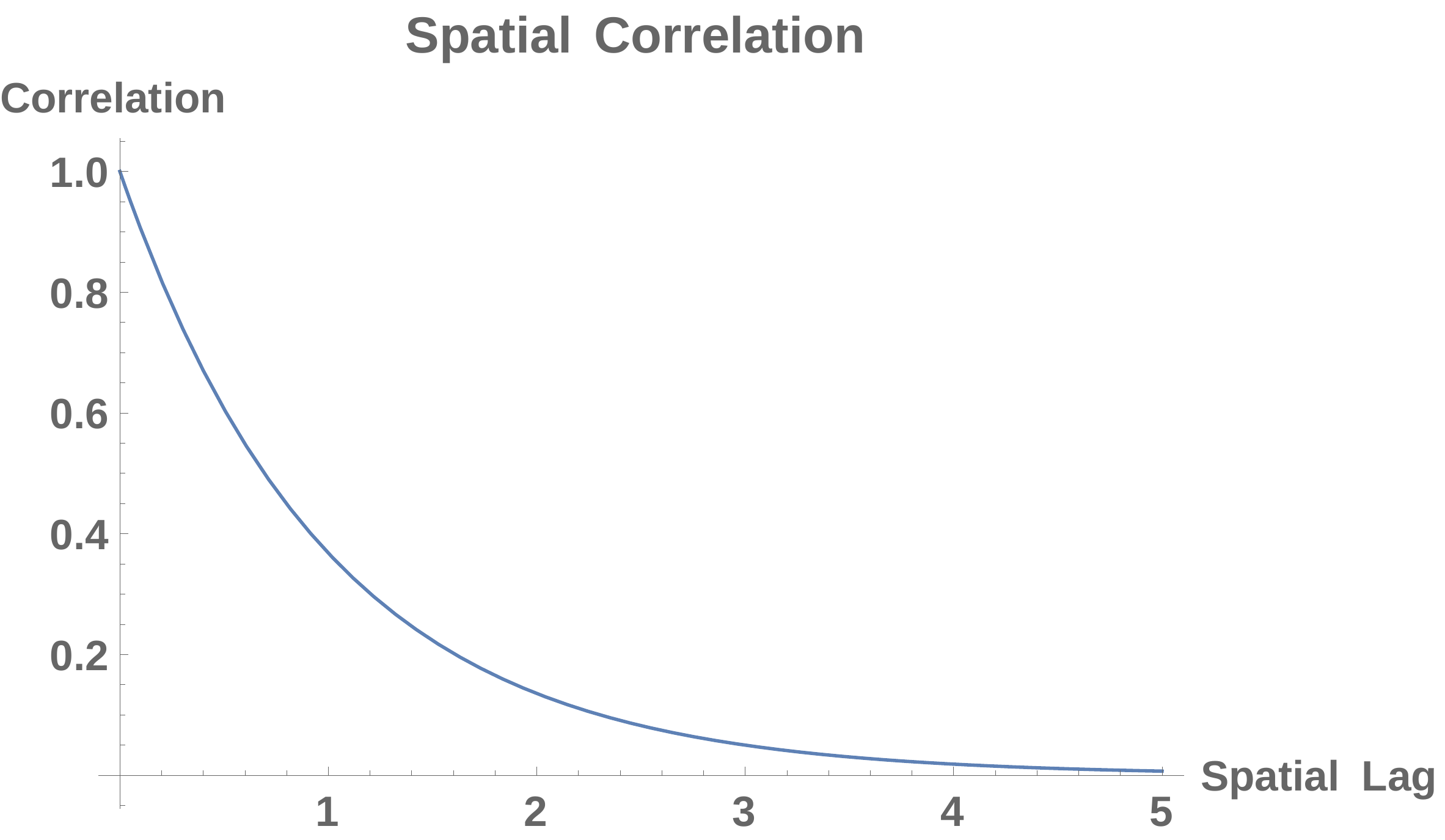}  \\ \includegraphics[scale=.28]{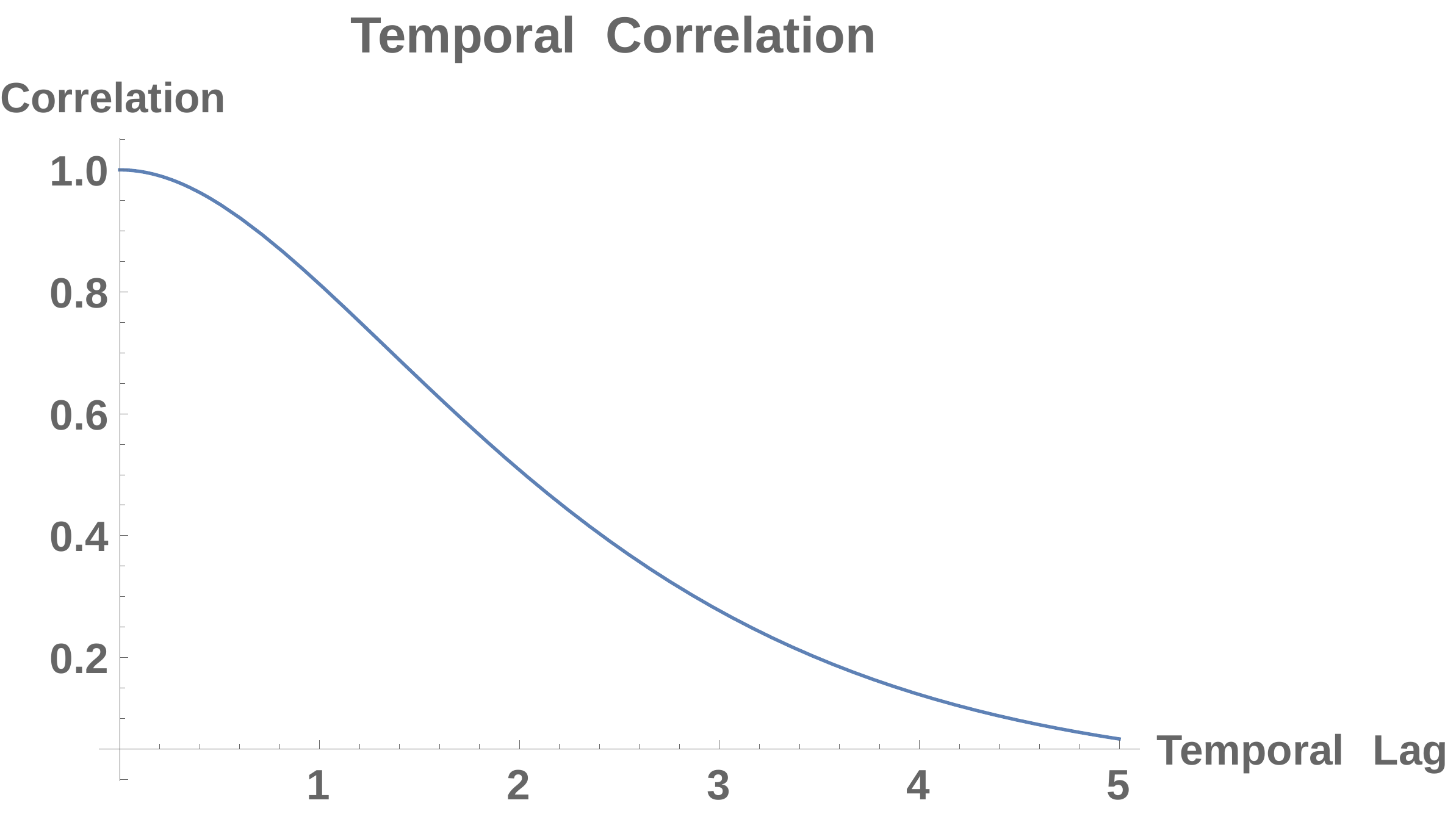}
\end{minipage}
\caption{Contour and marginal covariance plots for the Marginal Mat\'{e}rn model.  Effective smoothness in space is 1/2.  Effective smoothness in time is 2.}
\label{MarMatcovs}
\end{figure}
Figure \ref{MarMatcovs}(a) shows the dimple in this model is much more pronounced than that in Example 1.  Empirically, the dimple becomes less pronounced when the smoothnesses of the purely spatial and purely temporal processes are closer together, but the dimple does not disappear even when $\nu = \kappa$.  Explicit expressions for the half-spectrum of this model are given in the Appendix.

\end{example}

\section{Application}

We apply examples from the two classes of models presented in Section 3 to the Irish wind dataset first analyzed by \cite{Haslett}.  We compare fits from models in either class to similarly parameterized separable and non-separable models.  In making these comparisons, we focus on space-time interaction properties of the models we fit; hence, when possible, marginal spatial and temporal covariances and smoothnesses are fixed at common values across the models we fit.  We focus on space-time interaction to show the models we develop may be natural first choices to consider when analyzing space-time data.

The Irish wind dataset is a record of daily average wind speeds  collected at 12 sites in Ireland from 1961 to 1978.  This dataset has been analyzed several times previously by \cite{Haslett}, \cite{gneiting}, \cite{deluna} and \cite{regmondata} among others.  Like these authors, we deseasonalize the data by regressing wind-speed on a few (four) annual harmonics, and we discard data from the Rosslare site due to nonstationarities.  We also take square-roots to obtain more marginally Gaussian data, and we subtract station means of these deseasonalized and transformed data.  Finally, we divide each station's time series by its respective sample standard deviation.  The resulting dataset contains $6{,}574 \times 11 = 72{,}314$ observations with none missing.

There are subtleties in these data.  Most notably these data have been documented as having longer memory in time than covariance models with exponential tails can produce \citep{Haslett}.  Therefore, to make best comparisons within groups of models, we separate comparisons based on the temporal tail behavior (exponential tails versus algebraic tails) of the covariance models we consider.

The wind process also has properties that are not fully-symmetric \citep{gneiting}.  The primary model comparisons we make in this section involve fully-symmetric models because these provide the best comparisons to established models.  However, at the end of the application we fit a model that does not exhibit full-symmetry by using a simple model for the phase shift present in (\ref{hs}). 

\subsection{Regular monitoring data and model fitting}

Let $Z(s,t)$ be the daily average wind speed at a spatial location $s$ in Ireland observed on day $t$.  For these observations, we define models in 4-dimensional Euclidean space, $(s,t) \in \mathbb{R}^3 \times \mathbb{R}$, that we restrict to the actual domain, the sphere $\times$ time.  Since $Z(s,t)$ is observed at each spatial site, $s_i$, for each time point, $j$, we can write observational points in the form $(s_i,j)$ for all $i \in \{1, \hdots, 11\}$ and $j \in \{1, \hdots, 6574\}$.  Data of this form is called regular monitoring data \citep{regmondata}.

Fitting the models developed in Section 3 to regular monitoring data can be done by either computing covariances via a numerical evaluation of the integral in (\ref{hs}) or by viewing regular monitoring data as a multivariate time-series and using spectral methods.  We focus on spectral fitting in this paper; however, to quickly evaluate the integral in (\ref{hs}) at multiple spatial lags for a fixed temporal lag, discrete Fourier transforms can be used to speed computations.  We primarily use spectral fitting for convenience.  It should also be noted that $\mathbb{C}(s\delta(\omega))$ can be seen as the cross-spectra of spatially situated multivariate time series \citep{regmondata}.

Spectral fitting requires maximization of a multivariate Whittle likelihood.  The models in Section 3 are defined in continuous
space and time; hence, if we wish to use the multivariate Whittle likelihood, which estimates the spectrum of a discretely observed
time series, we will need to account for aliasing.
We use a simple aliasing correction. For the continuous model $f(\omega)\mathbb{C}(s \delta(\omega))$, we approximate a finite version of this function with $\sum_{j = -m}^{m} f(\omega + 2 j \pi) \mathbb{C}(s \delta(\omega+2 j \pi))$ for $-\pi < \omega < \pi$.  Note this sum is a truncation of the proper correction in which the sum goes from $-\infty$ to $\infty$.  We pick the truncation parameter $m = 50$ for convenience. For the models we fit, increasing $m$ beyond 50 appeared to have little effect on likelihoods and parameter estimates; however, in general, tail behavior of spectrums can vary tremendously; hence, this truncation should be determined carefully in specific applications.

\subsection{Results}

\subsubsection{Models with Exponentially Decaying Temporal Tails}

We fit the models in Examples 1 and 3 to these data.  For comparison, we fit a Mat\'{e}rn model as well as a separable model with exponential marginal covariances.  We also fit a model in \cite{creshua} that was found to best fit a different dataset of wind speeds.  Since the best-fitting model in \cite{creshua} is separable with squared exponential covariance in space and exponential covariance in time, the half-spectral form of this model is straightforward to obtain.  For the \cite{creshua} model, $f(\omega) \mathbb{C}(s\delta(\omega)) = \phi (\beta^2 + \omega^2)^{-1} \exp(-\alpha^2 \| s\|^2)$.  Note $\delta$ is constant in this model due to its separability.  \cite{creshua} fit their models using a weighted least squares method to match the model variogram to the empirical variogram.  In the model variogram they include a term related to $\| s \|$.  Since we subtract site means prior to analysis of the Irish wind data, we do not need to use an analog of this term in our models.

The only modification we make to all four exponential temporal tail models is to add a spatial nugget similar to \cite{gneiting} and \cite{regmondata}.  The half-spectral representations of our models therefore take the general form
\begin{equation}
f(\omega) \left[ \mathbb{C}(s \delta(\omega)) + \eta^2 \mathbb{I}_{\{s = 0\}} \right] \label{spacenug}
\end{equation}
where $\mathbb{I}_{\{\cdot\}}$ is an indicator function and $\eta^2 \geq 0$.

For models with flexibility in determining smoothness in time (Examples 1 and 2 and the Mat\'{e}rn models), temporal smoothness parameters, $\kappa$, were fixed at 1/2 to allow for better comparison of all models.  Note for the Mat\'{e}rn model, this also fixes the spatial smoothness.  To evaluate the assumption that $\kappa = 1/2$, we fit a univariate Mat\'{e}rn model to a subset of the time series at each station, and we found $\kappa = 1/2$ to be a reasonable estimate of temporal smoothness.

\begin{table}
\caption{\label{table1}Comparison of log-likelihoods of short memory in time models fitted to the Irish wind dataset.}
\centering
\fbox{%
\begin{tabular}{*{4}{l r r r }}
\multicolumn{4}{c}{Whittle Log-Likelihoods of Models with Exponential Temporal Tails} \\
\hline
& Log-Likelihood  & Diff. From Ex. 1 & \# Param. Fit \\
 \hline
Example 1 ($\kappa =  1/2$) & 20,318 & 0 & 5\\
Mat\'{e}rn (Example 1, $\kappa = \nu = 1/2$) & 20,199 &  119 & 4 \\
Example 3 ($\kappa = 1/2$) & 19,441 & 877 & 5 \\
Separable Exponential & 18,703 & 1,615 & 4 \\
\cite{creshua} & 18,378 & 1,940 & 4\\
\hline
\end{tabular}}
\end{table}

A comparison of the multivariate Whittle likelihoods is in Table \ref{table1}. Here, the model from Example 1 is shown to provide the best fit for these data.  That Example 1 outperforms a Mat\'{e}rn model is expected since Example 1 extends the Mat\'{e}rn class.  The added flexibility in specifying the spatial smoothness appears to add a small amount to the log-likelihood in this case.  With Example 1, the effective smoothness in space is estimated to be $0.4$; hence a modest improvement in log-likelihood is not surprising.  For processes estimated to have drastically different smoothnesses in space and time, we may  expect a larger improvement in using Example 1 over a Mat\'{e}rn model.

A look at the remaining fits shows smoothness and space-time interaction matter in terms of model fitting.  For example, the Mat\'{e}rn model compared to the separable exponential model shows an increase in 1,496 of the log-likelihood.  Both models have the same number of fitted parameters and the same smoothnesses in space and time; yet, the Mat\'{e}rn appears able to capture the space-time interaction of the Irish wind process better.  Comparing the separable exponential to the model in \cite{creshua}, we see modeling variations in space as analytic is perhaps unreasonable.  Comparing Example 3 to Example 1, we also see simply being able to fit marginal covariances may not directly lead to good models.  Nonetheless, Example 3 appears to capture space-time interaction in these data better than the separable models present in Table \ref{table1}.

\subsubsection{Models with Algebraic Temporal Tails}

\cite{gneiting} uses the following covariance model, denoted by $G$, to analyze these data:
\begin{equation}
G(s,t) = \phi (\beta | t |^\kappa + 1)^{-1} \exp\left( -\frac{\alpha \| s \|}{ (\beta | t |^\kappa + 1)^{\gamma/2}}\right) + \eta^2 (\beta | t |^\kappa + 1)^{-1} \mathbb{I}_{\{s = 0\}} \label{gmodel}
\end{equation}
where $\beta >  0$, $\kappa \in [0,2]$ and $\gamma \in [0,1]$.  The last term in (\ref{gmodel}) is the spatial nugget as also seen in (\ref{spacenug}).  The temporal marginal covariance $G(0,t)$ has algebraic tail behavior.  When $\kappa \leq 1$, $G$ is a classical long memory model in time.  In the case where $\kappa > 1$, $G(s,t)$ is not formally long memory in time since $G(0,t)$ will be integrable; however, the algebraic decay in $G(0,t)$ for large $t$ is much slower than the exponential decay present in the models in Table \ref{table1} and fits the Irish wind data better than exponentially decaying models in Section 4.2.1.

To fit $G$ using a multivariate Whittle likelihood method, the half-spectrum of $G$ is needed, but we do not know of a general closed form half-spectral representation of $G$ currently.  However, specific separable half-spectral forms of $G$ can be found.  For $\gamma = 0$, $G$ is separable and exponential in space.  For $\kappa = 1$, a closed form for the spectral density of $G(0,t)$ is readily calculated (see Appendix); hence, the half-spectrum of a separable $G$ that is long-memory in time can be specified as long as $\kappa = 1$ and $\gamma = 0$ are fixed.  Note that \cite{gneiting} estimated $\kappa = 1.544$.  If we instead let $\kappa = 2$, a closed form spectrum for $G(0,t)$ can also be found, but it is exponential and thus $G(0,t)$ for this model will not meet the condition in (\ref{cond}).  With $\kappa = 1$, $G(0,t)$ meets the condition in (\ref{cond}) since the spectrum of $G(0,t)$ denoted $f_G$ decays proportional to $\omega^{-2}$ (details in the Appendix).

Since $f_G$ decays asymptotically proportional to $\omega^{-2}$, we can plug $f_G$ into the class in (\ref{HSFORM2}) to obtain $K_{f_G}$, a model that satisfies the condition in (\ref{cond}).  Fitting both $K_{f_G}$ with a spatial nugget and $G$ with $\kappa =1$ and $\gamma = 0$ by their respective half-spectra, we obtain the results in Table \ref{table2}.  We fix $\nu = 1/2$ for $K_{f_G}$ to achieve the same smoothness in space for $K_{f_G}$ as exists for $G$ with $\gamma = 0$ and $\kappa = 1$.

\begin{table}
\caption{\label{table2}Comparison of log-likelihoods of longer memory models fitted to the Irish wind dataset.}
\centering
\fbox{%
\begin{tabular}{*{4}{l  r  r  r}}
\multicolumn{4}{c}{Whittle Log-Likelihoods of Models with Algebraic Temporal Tails} \\
\hline
 & Log-Likelihood  & Difference From $K_{f_G}$ & \# Parameters Fit \\
 \hline
$K_{f_G}$ ~~ ($\nu = 1/2$) & 21,655 & 0 & 4\\
$G$ ~~ ($\gamma = 0$, $\kappa = 1$) & 21,245 & 410 & 4 \\
\end{tabular}}
\end{table}

Table \ref{table2} shows again that compared to similar non-separable models, a separable model will not adequately fit these data.  This exercise shows further that the class of models presented in Section 3.1 is quite flexible and can be adapted straightforwardly to many different scenarios.  Comparing the log-likelihoods of the short memory models in Table \ref{table1} to the longer memory models in Table \ref{table2}, we find that $G$ with $\gamma = 0$, a separable model, to better fit these data than the non-separable models in Table \ref{table1}.  The difference between $G$ in Table \ref{table2} and Example 1 from Table \ref{table1} is 927 Whittle log-likelihood units.  This is not completely surprising given that long memory properties of these data are apparent in looking at marginal correlation plots \citep{Haslett}.   

As a final comparison, we fit $K_{f_G}$ and a non-separable version of $G$ by maximizing space-time likelihoods (as opposed to the frequency domain Whittle likelihoods).  The size of this dataset (72,314 observations) is quite difficult for space-time likelihood fitting of these models without making use of computational shortcuts.  Specifically, for these data, calculating a log-determinant and a solve involving a $72{,}314\times72{,}314$ covariance matrix is required to compute the space-time likelihood under the multivariate Gaussian assumption; however, even storing a matrix of this size in RAM will not be possible on most machines as it requires over 40GB of memory using standard precisions.  To perform this calculation, we make use of the regular monitoring structure of these data and the temporal stationarity of these models.  Sorting the data first by time and second by location leads to a Block-Toeplitz covariance matrix with blocks of size $11\times11$.  Using a Block-Toeplitz Levinson-type algorithm, one can compute the log-determinant and solve terms simultaneously with greatly reduced computational burden \citep{akaike}. 

The covariance model $G$ is given in its space-time representation in (\ref{gmodel}).  The space-time covariance function $K_{f_G}$ must be computed numerically from the half-spectrum specified by (\ref{HSFORM1}).  We approximate the covariance function $K_{f_G}$ via a finite approximation of the integral in (\ref{hs}).  As mentioned in Section 4.1, an approximation of the integral in (\ref{hs}) for many different time lags, $t$, can be computed quickly using the fast Fourier transform.

To ensure $G$ and $K_{f_G}$ are most comparable, we fix $\nu = 1/2$ for $K_{f_G}$ and $\kappa = 1$ for $G$.  With these parameter values, both models have the same smoothnesses in space and (up to range and scale parameters) identical marginal temporal covariance models.  For $G$, the separability parameter $\gamma$ is estimated from the data.  Results from fitting these models are in Table \ref{table3}.

\begin{table}
\caption{\label{table3}Comparison of log-likelihoods of longer memory models fitted to the Irish wind dataset. In contrast to Table \ref{table2}, $G$ is a non-separable model.}
\centering
\fbox{%
\begin{tabular}{*{10}{l  r  r  r}}
\multicolumn{4}{c}{Space-Time Log-Likelihoods of Models with Algebraic Temporal Tails} \\
\hline
 & Log-Likelihood  & Difference From $K_{f_G}$ & \# Parameters Fit\\
 \hline
$K_{f_G}$ ~~ ($\nu = 1/2$) & 21,471 & 0 & 4\\
$G$ ~~ ($\kappa = 1$) & 21,327 &  144 & 5\\
\hline
\end{tabular}}
\end{table}

Model $K_{f_G}$ fits the data better than $G$ while simultaneously requiring fewer estimated parameters.  The separability parameter $\gamma$ is estimated to be 0.5; hence $G$ is non-separable.  The marginal spatial and marginal temporal covariances in both models are quite similar; thus the slight edge $K_{f_G}$ appears to have over $G$ in this case may be due to the ability of $K_{f_G}$ to better capture space-time interactions in these data.  Comparing Table \ref{table3} to Table \ref{table2}, we see a difference in likelihoods for $K_{f_G}$.  Due to the approximations made in calculating both Whittle-type and space-time likelihoods, some discrepancy is expected.

\subsubsection{A Space-time Asymmetric Model}

Starting with $K_{f_G}$, we generate a space-time asymmetric model, $\tilde{K}_{f_G}$, by picking $\theta(\omega)\phi$ to be non-zero.  We use a simple linear translation model for the phase shift, giving $\theta(\omega) = \rho \omega$, where $\rho \in \mathbb{R}$.  The spatial direction of this shift is determined by $\phi$.  Since $\theta(\omega)$ is locally bounded, $\tilde{K}_{f_G}$ will satisfy (\ref{cond}).  \cite{gneiting} remarks on the east-west asymmetry in these data; hence, for simplicity, we fix $\phi \in \mathbb{R}^3$ to be a unit vector pointing in a direction consistent with an east-to-west direction over Ireland.  In principle, $\phi$ can be estimated from the data.  Lastly, for best comparison to $K_{f_G}$ in Table \ref{table2}, $\nu$ is again fixed at $1/2$ for $\tilde{K}_{f_G}$.

Using Whittle likelihood fitting, $\tilde{K}_{f_G}$ has a log-likelihood of 21,922.  Comparing to $K_{f_G}$ in Table \ref{table2}, we see a marked advantage to taking into account space-time asymmetry with these data.  The increase in Whittle log-likelihood is 267 units. The value $\rho$ was estimated to be positive; hence, $\tilde{K}_{f_G}$ is consistent with the simple measures of asymmetry used in \cite{gneiting}.

Lastly, it is important to note that in the context of modeling a process on the entire globe, this phase shift model will not be entirely sensible.  The vector $\rho \omega \phi$ indicates a linear phase shift through 3-dimensional Euclidean space.  Our data are on the globe however; therefore, an east-west phase shift is not everywhere accurately described by the single vector $\phi$.  For example, while $\phi$ is an east-to-west vector over Ireland, it is a west-to-east vector on the opposite side of the globe (for example, in Japan).  A better way to model phase shifts on the globe in the context of half-spectral models is given in \cite{regmondata}, wherein phase shifts are akin to rotations of the globe.  In this type of phase shift model, phases depend on differences in longitude and/or latitude and not a lag vector in $\mathbb{R}^3$.  Unfortunately, rotation based phase models do not fit exactly within the theory presented in Section 2.  Therefore, we have opted to use the locally accurate phase shift described by $\rho \omega \phi$.  Since Ireland covers only a small range of longitudes, using this simplification should not affect results in any substantial way.

\section{Discussion}


We developed two new classes of continuous space-time dependency models via half-spectra defined by the form $f(\omega) \mathbb{C}(s\delta(\omega)) \exp(i \theta(\omega) \phi)$.  Both classes of models allow for substantial flexibility in modeling the marginal processes defined by $K(s,0)$ and $K(0,t)$.  Moreover, both classes of models are non-separable.  These models are further theoretically validated by meeting at least one of two restrictions we developed in Section 2 to guide our model building.  These restrictions were developed through consideration of a natural condition posed in \cite{steinscreen}. Our model building focused on fully-symmetric models defined by $f(\omega) \mathbb{C}(s\delta(\omega))$; however, Theorem 1 in Section 2 established that inducing space-time asymmetry in these models using a non-zero $\theta(\omega)$ can be done naturally and easily.

We compared fully-symmetric examples from the new classes of models in this paper to separable models and to a non-separable model developed in \cite{gneiting} by fitting these models to the Irish wind dataset in Section 4.  Models from the class of half-spectra defined in (\ref{HSFORM1}) were shown to fit the Irish wind data better than all other models we examine in this paper in terms of likelihood.  \cite{Haslett} noted these data exhibit long-memory properties in time, and in Section 4, we showed the flexibility of the first class of models by fitting a model that is long-memory in time that also meets the condition (\ref{cond}).  We also fit an adaptation of this model using a non-zero $\theta(\omega)$ to illustrate the space-time asymmetry present in these data. Half-spectral models of the form in (\ref{HSFORM2}) therefore appear to be quite flexible and may be a good starting point for many model building purposes.

Space-time asymmetry was modeled as a phase shift using $\theta(\omega)$ in this paper.  Other methods exist that produce space-time asymmetry that may be of useful in half-spectral modeling.  Geometric transformations of space in particular may be used.  In such models, a spatial lag vector, $s$, may be substituted for a spatial lag shifted through time, $s + V t$, where $V$ is a vector indicating the magnitude and direction of the spatial shift.  The interpretation of this asymmetry is straightforward.  A model with this type of spatial shift is appropriate for physically moving processes.  An added benefit to using this type of asymmetry is that an adaptation of Theorem 1 exists in this setting: these types of geometric transformations will not affect whether or not a model satisfies (\ref{cond}).   Additional methods to produce space-time asymmetry may also be used in the half-spectral setting. Taking derivatives of a process in the manner of \cite{jun2} can lead to space-time asymmetries as well as spatial anisotropy. Non-parametric deformation as seen in \cite{deformation} can also produce asymmetry.

The half-spectral forms present in this paper serve as one step toward a more coherent melding of spatial and temporal methods in statistics. A primary avenue of future research may be to develop computationally efficient fitting procedures for half-spectral models on data that do not have the nice structure of non-missing regular monitoring data.  Adaptation of composite likelihood methods [\cite{lindsay}, \cite{vecchia} and \cite{steinvecchia}] may serve as one way forward as these methods allow for considerations of only subsets of the full dataset.

\section{Appendix}

In this appendix, we prove the Theorems and Restrictions presented in Section 2.  We give explicit forms for the half-spectrum of the model from Example 3 and of the marginal temporal spectrum, $f_G$ of $G$ with parameter $\kappa$ fixed at 1.  We also give details on the fitting procedures used in Section 4.

\subsection{Proofs of Theorem 1 Statement (i) and the Restrictions}

Proofs of Theorem 1 Statement (i) and Restrictions 1 and 2 can be found here.  Proof of Theorem 1 Statement (ii) has been put in Section 6.2 because its proof relies on two lemmas.

\begin{ths1}
Let $g(\lambda,\omega)$ be the full spectrum of a space-time covariance function where $\lambda$ is the spatial wavenumber and $\omega$ is the temporal frequency. Let $\theta(\cdot)$ be an odd function, and let $\phi$ be a unit vector in $\mathbb{R}^d$. The following statement holds:
\begin{itemize}
\item[(i)] Let $\theta(\cdot)$ be locally bounded.  The full spectrum $g(\lambda,\omega)$ satisfies (\ref{cond}) if and only if \newline $g(\lambda - \theta(\omega) \phi,\omega)$ satisfies (\ref{cond}).  
\end{itemize}
\end{ths1}

\begin{proof}
First, suppose $g(\lambda,\omega)$ satisfies (\ref{cond}).  If we can prove $\|(\lambda - \theta(\omega)\phi,\omega)\| \rightarrow \infty$ as $\|(\lambda,\omega)\| \rightarrow \infty$, we will have proved $g(\lambda - \theta(\omega)\phi,\omega)$ will satisfy (\ref{cond}).  When $\|(\lambda,\omega)\| \rightarrow \infty$, there are two cases to examine: $|\omega| \rightarrow \infty$ and the case where $\omega$ is bounded as $\| \lambda \|$ diverges.

If $|\omega| \rightarrow \infty$,  $\|(\lambda - \theta(\omega)\phi,\omega)\|$ diverges to infinity.  The remaining case, where $\omega$ is bounded as $\lambda$ grows, must be considered.  Since $\theta$ is locally bounded and $\omega$ is bounded, there exists some fixed $A$ such that $|\theta(\omega)| < A$ for all $\omega$ as $\| ( \lambda,\omega ) \| \rightarrow \infty$. We have
\begin{eqnarray}
\| \lambda - \theta(\omega) \phi \| &=& \sqrt{\lambda^T \lambda - 2 \theta(\omega) \lambda^T \phi + \theta(\omega)^2} \nonumber \\
& \geq & \sqrt{\| \lambda \|^2 - 2 \| \lambda \| A}. \nonumber
\end{eqnarray}
Thus, $\| \lambda \| \rightarrow \infty$ implies $\| \lambda - \theta(\omega) \phi \| \rightarrow \infty$.

The reverse direction may be proved by recognizing $g(\lambda,\omega)$ as simply a phase shifted version of $g(\lambda - \theta(\omega)\phi,\omega)$ using the valid phase shift function $-\theta(\omega)$.

\end{proof}



\begin{restriction}
Let the spectral representation of covariance function $K$ have the form $g(\lambda,\omega) = p(\omega) q(\lambda,\omega)$.  The following two statements hold for spectral densities with this parameterization.
\begin{itemize}
\item[(i)] If $g$ and $q$ satisfy the condition in (\ref{cond}), then $p$ must be constant in $\omega$.
\item[(ii)] If $g$ satisfies the condition in (\ref{cond}), then for every point $\omega_0$ such that $p(\omega_0)$ is finite and positive, the marginal spectrum $g^*_{\omega_0}(\lambda) = g(\lambda,\omega_0)$ must meet the condition in (\ref{cond}).
\end{itemize}
\end{restriction}

\begin{proof}
For both (i) and (ii), we prove the contrapositive of each statement.  First we prove (i). Assume $p$ is not constant and $q$ satisfies (\ref{cond}).  Then there exists some $\omega_1$ and $u_1$ and $R$ such that $p(\omega_1 + u_1)/p(\omega_1) = C > 1$ with $|u_1| < R$.  Let $\omega$ be fixed at $\omega_1$; hence $\|\lambda \| \rightarrow \infty$ in the limit in (\ref{cond}).  Taking the supremum in (\ref{cond}), we have
\begin{eqnarray}
\sup_{\|(u,v)\| < R} \left| \frac{p(\omega_1 + u)q(\lambda + v,\omega_1+u)}{p(\omega_1)q(\lambda,\omega_1)} - 1 \right| &\geq & \left| \frac{p(\omega_1 + u_1)q(\lambda,\omega_1+u_1)}{p(\omega_1)q(\lambda,\omega_1)} - 1 \right| \nonumber \\
&=& C \left|  \frac{q( \lambda,\omega_1+u_1)}{q( \lambda,\omega_1)} - \frac{1}{C} \right| \nonumber
\end{eqnarray}
From here, we can bound the limit
\begin{eqnarray}
\lim_{\|\lambda\| \rightarrow \infty} \sup_{\|(u,v)\| < R}  \left| \frac{p(\omega_1 + u)q(\lambda + v,\omega_1+u)}{p(\omega_1)q(\lambda,\omega_1)} - 1 \right| &\geq& C \lim_{\|\lambda\| \rightarrow \infty} \left|  \frac{q(\lambda, \omega_1+u_1)}{q( \lambda, \omega_1)} - \frac{1}{C} \right| \nonumber \\
& = & C -1  \nonumber
\end{eqnarray}
The last expression is strictly larger than 0; hence $g$ cannot satisfy the condition in (\ref{cond}).  By contraposition (i) is proved.

For case (ii), assume there exists $\omega_0$ such that $p(\omega_0)$ is finite and positive, and $q(\lambda,\omega_0)$ does not meet (\ref{cond}) marginally in $\lambda$. Fix $\omega = \omega_0$, and let $u = 0$. Since $0 < p(\omega_0) < \infty$,  $p(\omega_0)/p(\omega_0) = 1$. The condition in (\ref{cond}) on $g$ is therefore equivalent to the same condition on $q(\lambda,\omega_0)$, which is not satisfied by assumption. Statement (ii) follows again by contraposition. 
\end{proof}


\begin{restriction}
Let $g(\lambda,\omega) = p(\omega)q(\lambda/\delta(\omega))$, where $p(\omega)$ is positive and finite for all sufficiently large $|\omega|$, $q$ is a non-negative, continuous, integrable function that is not identically zero, and $\delta(\omega)$ is even with a well-defined limit as $|\omega| \rightarrow \infty$.  If $g$ satisfies (\ref{cond}),  then $\lim_{|\omega| \rightarrow \infty} \delta(\omega) = \infty$. 
\end{restriction}

\begin{proof}
Like Restriction \ref{t1}, we prove Restriction \ref{t2} by contraposition.  Assume $\lim_{|\omega| \rightarrow \infty} \delta(\omega) = E$, where $E < \infty$.  Take $\|(\omega,\lambda)\| \rightarrow \infty$ along a path for which $|\omega| \rightarrow \infty$ and $\lambda = \delta(\omega) v_0$, where $v_0 \in \mathbb{R}^d$ and $q(v_0) > 0$.  Since $q$ is integrable, there exists some $v_1 \neq 0$ and some $R$ such that $q(v_0 + v_1/E) \neq q(v_0)$ and $\|v_1\| < R$. The supremum in (\ref{cond}) can be bounded from below:
\begin{eqnarray}
\sup_{\|(u,v)\|<R} \left| \frac{p(\omega + u) q\left(\frac{1}{\delta(\omega+u)} \left( \delta(\omega)v_0 + v \right) \right)}{p(\omega) q\left(v_0 \right)} - 1 \right| & \geq & \sup_{\|(0,v)\| < R}\left| \frac{q\left(v_0+\frac{v}{\delta(\omega)}\right)}{q\left(v_0\right)} - 1 \right|  \nonumber \\
& \geq & \left| \frac{q\left(v_0+\frac{v_1}{\delta(\omega)}\right)}{q\left(v_0\right)} - 1 \right| \nonumber
\end{eqnarray}
Since $q$ is continuous and $\lim_{|\omega| \rightarrow \infty} \delta(\omega) = E$, the limit of the supremum in (\ref{cond}) is bounded away from zero: $\lim_{\|(\lambda,\omega)\| \rightarrow \infty} \sup_{\|(u,v)\|<R} |g(\lambda + v,\omega + u)/g(\lambda,\omega) - 1| > 0$. 
\end{proof}

\subsection{Proof of Theorem 1, Statement (ii)}

The proof of Theorem 1, Statement (ii) relies on the following two lemmas.

\begin{lemma}
Let $g(\xi)$, $\xi \in \mathbb{R}^{d+1}$ be a non-negative, integrable function.  If $g(\xi)$ satisfies (\ref{cond}), then $\lim_{\|\xi\|\rightarrow\infty} g(\xi) = 0$. \label{Lemma1}
\end{lemma}

\begin{proof}
We prove this statement using contradiction.  Suppose $\lim \sup_{x \rightarrow \infty} g(\xi_0(x)) \neq 0$ for a certain path indexed by a parameter, $x$, $\xi_0(x):\mathbb{R}^+ \rightarrow \mathbb{R}^{d+1}$ such that $\| \xi_0(x) \| \rightarrow \infty$ as $x$ grows large.  Therefore, there exists some $\varepsilon > 0$ such that for each $i \in \{1,2, \hdots \}$, there exist some $x_i \geq i$ where $g(\xi_0(x_i)) > \varepsilon$.  Since $\| \xi_0(x_i)\| \rightarrow \infty$ as $i \rightarrow \infty$, we can select a subsequence, written $y_i$ for $i \in \{1, 2, \hdots \}$ such that $\| \xi_0(y_i) - \xi_0(y_j) \| > 2$ for $i \neq j$.

Now consider $\sum_{i=1}^\infty \int_{B(\xi_0(y_i),1)} g(\xi) d\xi$, where $B(\xi_0(y_i),1)$ indicates the ball of radius 1 centered at point $\xi_0(y_i)$.  Because the radius of these balls is less than or equal to half the distances between each center point, $\xi_0(y_i)$, these balls do not overlap.  Thus the sum of their volume will lower bound the full integral $\int_{\mathbb{R}^{d+1}} g(\xi) d\xi$.  To complete this proof by contradiction, we show $\sum_{i=1}^\infty \int_{B(\xi_0(y_i),1)} g(\xi) d\xi$ explodes by showing terms at sufficiently large values of $i$ will have a positive number as a lower bound.

We have $g(\xi_0(y_i)) > \varepsilon > 0$ for each $i$.  And we have $g(\xi)$ satisfies (\ref{cond}).  For a radius, $R$, set $R=1$, and pick $\delta = 1/2$.  Since $g(\xi)$ satisfies (\ref{cond}), it holds that there exists an $I \in \mathbb{R}$ such that for all $i > I$,
\[ \frac{1}{2} < \frac{g(\xi_0(y_i) + U)}{g(\xi_0(y_i))} < \frac{3}{2} \]
for all points such that $\|U\| < 1$.  Since $g(\xi_0(y_i)) > \varepsilon$, we find a lower bound for all points in the unit ball centered at $\xi_0(y_i)$, $\inf_{\xi \in B(\xi_0(y_i),1)} g(\xi) \geq \varepsilon/2$.
Therefore, $\int_{B(\xi_0(y_i),1)} g(\xi) d\xi \geq \varepsilon\pi^{(d+1)/2}/((d+1) \Gamma((d+1)/2)$ for each $i > I$.  Thus $\int_{\mathbb{R}^{d+1}} g(\xi) d\xi$ diverges, but this contradicts the integrability assumption. 

\end{proof}

\begin{lemma}
Let $g(\xi)$, $\xi \in \mathbb{R}^{d+1}$, be a strictly positive, integrable function.  Let $\eta(\cdot)$, $\xi_0(\cdot)$ and $U_0(\cdot)$ be a triplet, $(\eta(\cdot),\xi_0(\cdot), U_0(\cdot))$, where $\eta:\mathbb{R}^{d+1} \rightarrow \mathbb{R}^{d+1}$, $\xi_0:\mathbb{R} \rightarrow \mathbb{R}^{d+1}$ and $U_0: \mathbb{R} \rightarrow \mathbb{R}^{d+1}$ such that $(\eta(\cdot),\xi_0(\cdot), U_0(\cdot))$ has the following properties:
\begin{itemize}
\item[(i)] $\|\xi_0(x)\| \rightarrow \infty$ as $x \rightarrow \infty$.
\item[(ii)] $\| \eta(\xi_0(x))\|$ is bounded.
\item[(iii)] For any $C > 0$ and for all $x \in \mathbb{R}^+$, there exists a $Y \in \mathbb{R}$ dependent only on $C$ such that for all $y > Y$,
\[\| \eta(\xi_0(x)+U_0(y) )\| > C. \]
\end{itemize}
If $g(\xi)$ satisfies (\ref{cond}), then the function $f(\xi) = g(\eta(\xi))$ cannot satisfy (\ref{cond}).
\end{lemma}

\begin{proof}
Loosely, the properties of the triplet, $(\eta(\cdot),\xi_0(\cdot),U_0(\cdot))$, ensure that we can pick a point, $y$, large enough such that $\eta(\xi(x) + U_0(y))$ is sufficiently far from $\eta(\xi_0(x))$ for all $x$.  This, together with Lemma \ref{Lemma1}, prevents (\ref{cond}) from being satisfied for $f(\xi)$.

More formally, consider the supremum in condition (\ref{cond}) for $f(\xi_0(x))$ and for any $R \in \mathbb{R}$.
\[ \sup_{\|U\| < R} \left| \frac{f(\xi_0(x) + U)}{f(\xi_0(x))} - 1 \right |  \geq  \left| \frac{f(\xi_0(x) + U_0(y))}{f(\xi_0(x))} - 1 \right |  \]
for any $y$ such that $\| U_0(y) \| < R$.  Since $R$ can be arbitrarily large, we know this condition can be satisfied.  Substituting in $g$ on the right hand side, we find
\[ \sup_{\|U\| < R} \left| \frac{f(\xi_0(x) + U)}{f(\xi_0(x))} - 1 \right |  \geq  \left| \frac{g(\eta(\xi_0(x) + U_0(y)))}{g(\eta(\xi_0(x)))} - 1 \right |.  \]
We know $\| \eta(\xi_0(x)) \|$ is bounded and that $g$ is strictly positive; hence, there exists some $\varepsilon > 0$ such that $g( \eta(\xi_0(x))) > \varepsilon$ for all $x$.

By Lemma \ref{Lemma1} and property (iii) we also know $y$ (and thus $R$) can be chosen large enough such that $g(\eta(\xi_0(x) + U_0(y))) < \varepsilon/2$ for all $x$. The limit of $\sup_{\|U\| < R} \left| \frac{f(\xi_0(x) + U)}{f(\xi_0(x))} - 1 \right |$  as $x\rightarrow\infty$ must therefore be bounded below by 1/2 and thus cannot equal 0. Since $\| \xi_0(x) \| \rightarrow \infty$ as $x \rightarrow \infty$, we have found a path, $\xi_0(x)$, that violates condition (\ref{cond}) for $f(\xi)$. 
\end{proof}

\begin{ths2}
Let $g(\lambda,\omega)$ be the full spectrum of a space-time covariance function where $\lambda$ is the spatial wavenumber and $\omega$ is the temporal frequency. Let $\theta(\cdot)$ be an odd function, and let $\phi$ be a unit vector in $\mathbb{R}^d$. Let also $g(\lambda,\omega)$ be strictly positive, and let $g(\lambda,\omega)$ satisfy (\ref{cond}). If $g(\lambda - \theta(\omega) \phi,\omega)$ satisfies (\ref{cond}), then $\theta(\omega)$ must be locally bounded.
\end{ths2}

\begin{proof} We prove this statement using contradiction.  Assume $\theta(\omega)$ is not locally bounded.  The statement then follows almost directly from Lemma 2. To meet the conditions of Lemma 2, we need to define a triplet, $(\eta(\cdot),\xi_0(\cdot), U_0(\cdot))$ that has the properties listed in Lemma 2.  

First we define $\xi_0(x)$. Since $\theta(\omega)$ is not locally bounded, there exists a sequence, \newline $\omega_1, \omega_2, \hdots$, and a bound, $A$, such that $|\omega_i| < A$ for all $i$ and $\| \theta(\omega)\| \rightarrow \infty$ as $i \rightarrow \infty$.  Define $\omega_0(x) = \omega_i$ if $i-1 < x \leq i$ and $\omega_0 = 0$ otherwise.  We have $|\omega_0(x)| < A$ for all $x \in \mathbb{R}$, and we have that $\lim_{x\rightarrow\infty} \|\theta(\omega_0(x))\| \rightarrow \infty$.  Define $\xi_0(x) = (\theta(\omega_0(x))\phi,\omega_0(x))$.  Thus $\|\xi_0(x)\| \rightarrow \infty$ as $x \rightarrow \infty$.

Next, define $\eta(\lambda,\omega) = (\lambda - \theta(\omega)\phi ,\omega)$.  Consider $\eta(\xi_0(x)) = (  \theta(\omega_0(x))\phi - \theta(\omega_0(x))\phi, \omega_0(x)) = (0,\omega_0(x))$.  We know $| \omega_0(x) | < A$; therefore, $\|\eta(\xi_0(x)) \|$ is bounded.

Finally we define $U_0(y) =  (y \phi,0)$.  We prove the triplet has property (iii) in Lemma 2. Let $C > 0$. Consider $\eta(\xi_0(x)+U_0(y) ) = (U_0(y),\omega_0(x)) = (y \phi, \omega_0(x))$.  Therefore, $\| \eta(\xi_0(x)+U_0(y) ) \| \geq \| y \phi \| = y$.  Let $Y = C$.  If $y > Y$, it follows that $\| \eta(\xi_0(x)+U_0(y) ) \| > C$ for all $x > 0$.

By assumption, $g(\xi) = g(\lambda,\omega)$ satisfies (\ref{cond}), and above, we established the triplet, $(\eta(\cdot),\xi_0(\cdot), U_0(\cdot))$, has the properties in Lemma 2.  Therefore, by Lemma 2, we know $g(\eta(\xi)) = g(\lambda - \theta(\omega)\phi ,\omega)$ cannot satisfy (\ref{cond}). By contradiction, $\theta(\omega)$ must be locally bounded.
\end{proof}

\subsection{Half-Spectrum for Example 3}

Writing down the half-spectrum in Example 3 requires expressions for the quantile function (inverse distribution function) of $|A| \sim1/\sqrt{2 \chi^2_{2\nu}}$ and the distribution function of the density $f_B \propto (\beta^2 + \omega^2)^{-(\kappa + 1/2)}$. Write $P(s,x) = \gamma(s,x)/\Gamma(s)$ as the regularized gamma function, where $\gamma$ is the lower-incomplete gamma function \citep{nist}.  The distribution function $F_{\chi^2_{2\nu}}(x;\nu)$ is given by $P(\nu, x/2)$.  Hence $F_{\chi^2_{2\nu}}^{-1}(x;\nu) = 2 P^{-1}(\nu,x)$, and therefore, $F_{|A|}^{-1} = 1/\sqrt{2 F_{\chi^2_{2\nu}}^{-1}(x;\nu)}$. The inverse distribution function $F_A^{-1}$ can be found directly as
\[ F_A^{-1}(x;\nu) = \left\lbrace \begin{array}{c c} -  1/\sqrt{4  P^{-1}(\nu,1-2x)}; & x< 1/2 \\
 1/\sqrt{4  P^{-1}(\nu,2x-1)}; & x \geq 1/2
\end{array}
\right. .
\]

To find $F_B$ it is convenient to view $f_B$ as a special case of the hypergeometric function, $_2F_1$ \citep{nist}. We write $f_B \propto \beta^{-(2\kappa+1)}~_2F_1(1, \kappa + 1/2;1;-\omega^2/\beta^2)$ \citep{nist}.  For $|\omega/\beta| < 1$, indefinite integration can be carried out directly using the series representation of the hypergeometric function to find $\int f_B d\omega \propto \omega/\beta^{2k+1} ~_2F_1(1/2,\kappa + 1/2; 3/2; -\omega^2/\beta^2)$ plus some constant.  Hence by analytic continuation, 
\[ F_B(\omega;\beta,\kappa) = \frac{1}{2} + \frac{\Gamma(\kappa+1/2)}{\Gamma(\kappa) \sqrt{\pi}} \frac{\omega}{\beta} ~_2F_1\left(\frac{1}{2},\frac{1}{2} + \kappa;\frac{3}{2}; -\frac{\omega^2}{\beta^2} \right)
\]
can be defined for the values of $\omega$ outside this range \citep{nist}.  The constant in $f_B$ and $F_B$ can be obtained in \cite[3.241.4]{gradsh}.

Define $\delta(\omega;\beta;\nu;\kappa) = |F_A^{-1}(F_B(\omega;\beta,\kappa);\nu)|$.  With $\mathbb{C}(s;\alpha) = \phi \exp(-\alpha^2 s^2)$, we find the half spectral form of Example 3,

\[ f(\omega)\mathbb{C}(s \delta(\omega)) = \phi (\beta^2 + \omega^2)^{-(\kappa + 1/2)} \exp(-\alpha^2 s^2\delta^2(\omega;\beta;\nu;\kappa)). \]

\subsection{Half-Spectrum of $K_{f_G}$}

To find the half-spectrum of $K_{f_G}$, we need only find the spectrum of the temporal covariance function $G(0,t)|_{\kappa = 1} = \phi (\beta |u| + 1)^{-1}$ written without the spatial nugget in (\ref{gmodel}) and equivalently in (\ref{spacenug}). \cite[p. 8]{bateman} gives the calculation as a cosine transform,
\[ f_G(\omega) = \phi/\beta \left[ \pi \sin\left(\frac{|\omega|}{\beta}\right) - 2 ~\mbox{Si}\left(\frac{|\omega|}{\beta} \right) \sin\left(\frac{|\omega|}{\beta}\right)  - 2~ \mbox{Ci}\left(\frac{|\omega|}{\beta} \right) \cos\left(\frac{|\omega|}{\beta}\right)  \right], \]
where $\mbox{Si}$ and $\mbox{Ci}$ are the Sine and Cosine integrals defined in \cite{nist}.  The tail of this spectrum decays asymptotically proportional to $\omega^{-2}$ \citep{nist}.  With $f_G$, the half-spectrum of $K_{f_G}$ is obtained by plugging $f_G$ into (\ref{HSFORM1}). 

\subsection{Fitting Procedures}

R is used for fitting of all models using the general optimizer \texttt{nlm} with default optimization parameters.  Numerical hessians are calculated in this implementation of \texttt{nlm}.  In the previous two sections of this Appendix, we showed how half-spectra can be calculated provided functions $_2F_1$, Sine and Cosine integrals and the regularized gamma function, $P$, can be calculated.  For $P$, we used the $\chi^2_{2\nu}$ quantile function included with $R$.  We use the \texttt{gsl} package for the Sine and Cosine integral functions \citep{gsl}, and we used the \texttt{hypergeo} package for $_2F_1$ \citep{hankinlee}.  The Block-Toeplitz solver for space-time fitting was adapted from a MATLAB function by \cite{blocklevinson}.


\small 
\bibliographystyle{agsm}
\bibliography{bibsave2}

\end{document}